\documentclass[12pt]{article}
\usepackage[margin=1in]{geometry}

\usepackage[T1]{fontenc}
\usepackage[utf8]{inputenc}
\usepackage{mathpazo}
\usepackage{microtype}
\usepackage{amssymb,amsmath,amsfonts,eurosym,ulem,graphicx,setspace,pdflscape,array,tikz,amsthm,enumerate,pifont,wasysym,xcolor}
\usepackage[labelfont=bf, font=small, labelsep=period]{caption}
\usepackage{ragged2e}
\usepackage{textcomp}
\usepackage{overpic}
\usepackage{pict2e}
\usepackage{pgfplots}
\usepackage{pgfplotstable}
\usepgfplotslibrary{fillbetween} %
\usepackage{booktabs}
\usepackage{appendix}
\usepackage{subcaption}
\usepackage{algpseudocode}
\usepackage{algcompatible}
\usepackage{algorithm}
\usepackage{longtable}
\usepackage{float}
\usepackage[round]{natbib}
\usepackage{multibib}
\usepackage{etoc}
\newcites{supp}{References}
\makeatletter
\newcommand\citeauthorsupp[1]{{\let\@auxout\@auxoutsupp\citeauthor{#1}}}
\newcommand\citeyearsupp[1]{{\let\@auxout\@auxoutsupp\citeyear{#1}}}
\makeatother
\definecolor{mydarkblue}{rgb}{0.0,0.0,0.4}
\definecolor{figred}{rgb}{0.6,0.08,0.08}
\usepackage[colorlinks,citecolor=mydarkblue,linkcolor=mydarkblue,urlcolor=mydarkblue]{hyperref}

\graphicspath{{figures/}}

\pgfplotstableset{
    col sep=comma,
    header=has colnames,
    string type,
    trim cells=true
}

\pgfplotstableset{
    my style/.style={
        columns/State/.style={string type},
        columns/LP/.style={column type=c},
        columns/VAR/.style={column type=c},
        every head row/.style={output empty row},
        every last row/.style={after row=\midrule},
    }
}

\makeatletter
\renewcommand\paragraph{\@startsection{paragraph}{4}{\z@}%
  {-3.25ex\@plus -1ex \@minus -.2ex}%
  {-1.2em}%
  {\normalfont\normalsize\itshape}}
\makeatother

\theoremstyle{plain}

\theoremstyle{definition}

\newtheorem{remark}{Remark}

\theoremstyle{plain}
\newtheorem{prop}{Proposition}
\newtheorem{lem}{Lemma}

\theoremstyle{definition}
\newtheorem*{innerexample}{Example}
\newcounter{exampleinvocations}
\newenvironment{example}{%
  \stepcounter{exampleinvocations}%
  \ifnum\value{exampleinvocations}>1
    \begin{innerexample}[continued]%
  \else
    \begin{innerexample}%
  \fi
}{%
  \end{innerexample}%
}
\newtheorem{specification}{Specification}
\newtheorem{assumpLP}{Assumption}
\newtheorem{assumpsLP}{Assumption}

\newtheorem{assumpexo}{Assumption}

\newtheorem{assumpexoVAR}{Assumption}

\newtheorem{assumpexoLP}{Assumption}

\newcommand\E{\mathbb{E}}
\newcommand\V{\mathbb{V}}
\newcommand\I{\mathbb{I}}

\newcommand{\ind}{\mathrel{\perp\!\!\! \!\perp}}

\usetikzlibrary{arrows.meta}
\tikzset{
  timeline/.style={-Latex, thick},
  forecastarrow/.style={->, very thick, draw=figred, shorten >=2pt, shorten <=2pt},
  param/.style={circle, draw, fill=white, minimum size=11pt, inner sep=0pt, font=\scriptsize},
}

\newcommand{\legend}[1]{{\footnotesize #1}}
\newcommand{\notes}[1]{\vspace{2.5em}\begin{minipage}{\linewidth}\footnotesize\noindent\textit{Notes:} #1\end{minipage}}

\begin{document}
\etocdepthtag.toc{main}

\newcommand{\corrauthsymbol}{\textsuperscript{\textdagger}}

\title{\large State-Dependent Local Projections and Vector Autoregressions\thanks{A previous version of this paper was titled `When and Why State-Dependent Local Projections Work'. I am very grateful to Serena Ng for continuous advice and feedback on this project. I also thank Richard Crump, Sílvia Gonçalves and Pablo Mones for valuable feedback on the draft and participants of the Columbia Econometrics Colloquium, the NBER-NSF Time Series Conference, the NOeG Winter Workshop and the Rome Conference on Macroeconometrics and Time Series for helpful comments.}}
\author{Valentin Winkler\corrauthsymbol}
\date{November 2026}

\maketitle

\begingroup
\renewcommand\thefootnote{\fnsymbol{footnote}}
\footnotetext[2]{Department of Economics, Columbia University, Email: \texttt{valentin.winkler@columbia.edu}.}
\endgroup

\bigskip

\begin{abstract}
  This paper studies the relationship between state-dependent local projections (LPs) and vector autoregressions (VARs). I show nonparametrically that state-dependent LPs estimate the projection of a convex weighted average of causal effects on the state, including continuous and multidimensional states. In contrast, conventional state-dependent VARs need not have a causal interpretation, even when they model state transitions correctly. A new backshifted VAR, conditioning every prediction step on the pre-shock state, asymptotically matches the LP estimand up to its lag order.
\bigskip


\end{abstract}

\clearpage

\section{Introduction}\label{s1}

In macroeconomics, the effect of an observed shock $X_t$ on a future outcome $Y_{t+h}$ is commonly estimated by running a local projection \citep[LP, ][]{Jorda:05} of the form

\begin{equation*}
    Y_{t+h} = X_t \beta^h + \text{error}_{h,t+h}.
\end{equation*}

A popular specification allows for the effect of $X_t$ on $Y_{t+h}$ to depend on the state of the economy.\footnote{Since this paper only studies asymptotic properties, it abstracts from control variables that are included to improve finite-sample performance. If controls are used for identification, assume that they have already been projected out using the Frisch--Waugh--Lovell theorem.} Such state-dependent LPs can be estimated as

\begin{equation}\label{eq:state_dep_simple}
  Y_{t+h} =  X_t \beta^h_0 + S_{t-1} X_t \beta_1^h + \text{error}_{h,t+h},
\end{equation}

where $S_{t-1}$ is a lagged, observed state variable that can be continuous or binary.\footnote{Most applied papers use a lagged state, even though some interact with a contemporaneous state $S_t$ (see Supplemental Appendix S.1).} If the interaction term $\beta_1^h$ is found to be non-zero, the effect of interest is commonly judged to be state-dependent.

Despite their popularity, state-dependent LPs are less well understood than unconditional LPs. The same is true for the leading alternative, state-dependent vector autoregressions (VARs). In particular, it is unclear how these methods relate to each other. When discussing state-dependent LPs, \citet[][p.~177]{Jorda:26} noted:

\begin{quotation}
  This is just an example of how quickly the departure between LPs and
VARs takes place because, to my knowledge, there is no obvious equivalent way to specify a VAR that would deliver the counterfactual interpretation of LPs in general.
\end{quotation}
  
Early influential studies using state-dependent LPs already discussed their relationship to state-dependent VARs. \cite{Ramey:18} find that fiscal multipliers in times of economic slack are roughly the same as in normal times. This contrasts with \cite{Auerbach:12}, who find elevated fiscal multipliers in recessions. Carefully replicating the earlier study, \cite{Ramey:18} largely attribute this disagreement to the difference between state-dependent LPs, which they use, and \citeauthor{Auerbach:12}'s (\citeyear{Auerbach:12}) VAR estimator. While LPs average over future transitions between states, the VAR estimator implicitly assumes that the economy will remain in the pre-shock state for multiple years. This difference was already described in \cite{Auerbach:13}. My paper shows that the disagreement between conventional state-dependent VARs and LPs runs deeper than that. Correctly adjusting for the transition law of the state does \textit{not} make it disappear. 
  
This paper explains why state-dependent LPs and conventional state-dependent VARs generally disagree and how a VAR-based estimator reconciles them. It makes three contributions.

First, I show that the asymptotic equivalence of unconditional LPs and VARs \citep{Plagborg-Moller:21} breaks down in the state-dependent case. Figure \ref{fig:lp_var} illustrates why: LPs directly predict $Y_{t+h}$ while conditioning on the lagged state $S_{t-1}$. VARs iterate one-step predictions of $(X_{t+k},Y_{t+k})$ while conditioning on $S_{t+k-1}$, so the projection operator changes for every prediction step. Since many different operators are combined, the equivalence between direct and iterative projections (the tower property of projections) does not hold.

\begin{figure}[t]
    \centering 
    \caption{Prediction Steps and Projection Operators of State-Dependent LPs and VARs ($h=3$).}\label{fig:lp_var}
  
    \begin{subfigure}[b]{0.3\textwidth}
      \centering
      \vspace{2mm}
      \begin{tikzpicture}[
          >=latex,
          timeline/.style={-},
          steparrow/.style={->, line width=1.2pt, draw=figred},
          param/.style={circle, draw=black!60, fill=white,
                        inner sep=1.8pt, font=\scriptsize}
        ]
  
        \def\h{3}
        \def\xstep{1.4}
        \def\ymax{0.8}
  
        \draw[timeline,->] (0,0) -- (\h*\xstep+0.4,0);
  
        \foreach \i/\lbl in {0/{t},1/{t+1},2/{t+2},3/{t+3}} {
          \draw (\i*\xstep,0.08) -- (\i*\xstep,-0.08);
          \node[below=3pt] at (\i*\xstep,0) {$\lbl$};
        }
  
        \draw[steparrow]
          (0,0)
          .. controls (\xstep,\ymax) and (2*\xstep,\ymax) ..
          (\h*\xstep,0)
          node[param, midway, above, yshift=3pt] {$P_{S_{t-1}}$};
  
      \end{tikzpicture}
      \subcaption{Local Projections}
    \end{subfigure}
    \hfill  \begin{subfigure}[b]{0.3\textwidth}
      \centering
      \begin{tikzpicture}[
          >=latex,
          timeline/.style={-},
          steparrow/.style={->, line width=1.2pt, draw=figred},
          param/.style={circle, draw=black!60, fill=white,
                        inner sep=1.8pt, font=\scriptsize}
        ]
  
        \def\h{3}
        \def\xstep{1.4}
        \def\ymax{0.8}
  
        \draw[timeline,->] (0,0) -- (\h*\xstep+0.4,0);
  
        \foreach \i/\lbl in {0/{t},1/{t+1},2/{t+2},3/{t+3}} {
          \draw (\i*\xstep,0.08) -- (\i*\xstep,-0.08);
          \node[below=3pt] at (\i*\xstep,0) {$\lbl$};
        }
  
        \draw[steparrow]
          (0,0)
          .. controls (0.3*\xstep,\ymax) and (0.7*\xstep,\ymax) ..
          (\xstep,0)
          node[param, midway, above, yshift=3pt] {$P_{S_{t+0}}$};
  
        \draw[steparrow]
          (\xstep,0)
          .. controls (1.3*\xstep,\ymax) and (1.7*\xstep,\ymax) ..
          (2*\xstep,0)
          node[param, midway, above, yshift=3pt] {$P_{S_{t+1}}$};
  
        \draw[steparrow]
          (2*\xstep,0)
          .. controls (2.3*\xstep,\ymax) and (2.7*\xstep,\ymax) ..
          (3*\xstep,0)
          node[param, midway, above, yshift=3pt] {$P_{S_{t+2}}$};
  
      \end{tikzpicture}
      \subcaption{VAR (Moving State)}
    \end{subfigure}
    \hfill\begin{subfigure}[b]{0.3\textwidth}
      \centering
      \begin{tikzpicture}[
          >=latex,
          timeline/.style={-},
          steparrow/.style={->, line width=1.2pt, draw=figred},
          param/.style={circle, draw=black!60, fill=white,
                        inner sep=1.8pt, font=\scriptsize}
        ]
  
        \def\h{3}
        \def\xstep{1.4}
        \def\ymax{0.8}
  
        \draw[timeline,->] (0,0) -- (\h*\xstep+0.4,0);
  
        \foreach \i/\lbl in {0/{t},1/{t+1},2/{t+2},3/{t+3}} {
          \draw (\i*\xstep,0.08) -- (\i*\xstep,-0.08);
          \node[below=3pt] at (\i*\xstep,0) {$\lbl$};
        }
  
        \draw[steparrow]
          (0,0)
          .. controls (0.3*\xstep,\ymax) and (0.7*\xstep,\ymax) ..
          (\xstep,0)
          node[param, midway, above, yshift=3pt] {$P_{S_{t-1}}$};
  
        \draw[steparrow]
          (\xstep,0)
          .. controls (1.3*\xstep,\ymax) and (1.7*\xstep,\ymax) ..
          (2*\xstep,0)
          node[param, midway, above, yshift=3pt] {$P_{S_{t-1}}$};
  
        \draw[steparrow]
          (2*\xstep,0)
          .. controls (2.3*\xstep,\ymax) and (2.7*\xstep,\ymax) ..
          (3*\xstep,0)
          node[param, midway, above, yshift=3pt] {$P_{S_{t-1}}$};
  
      \end{tikzpicture}
      \subcaption{VAR (Backshifted)}
    \end{subfigure}
  \end{figure}  

Second, I propose a VAR-based impulse response estimator that asymptotically matches the state-dependent LP estimand at all horizons up to its lag order, just as unconditional LPs and VARs do \citep{Plagborg-Moller:21}. Instead of letting the operator vary at every projection step, this estimator conditions every one-step projection on the pre-shock state $S_{t-1}$. Panel (c) of Figure \ref{fig:lp_var} visualizes this idea. The estimator is implemented by running $h+1$ state-dependent VARs. The first VAR conditions on $S_{t-1}$, the second on $S_{t-2}$, and so on. The slope coefficients are then iterated in the same order to form the impulse response estimate. Because every successive VAR shifts back the state variable by one time period, I call this estimator the `backshifted' VAR.
  
The backshifted VAR shows how the relationship between unconditional LPs and VARs carries over to the state-dependent case. Simulations show that the new estimator has lower variance than state-dependent LPs at longer horizons, thereby reintroducing the finite-sample bias-variance tradeoff highlighted in the unconditional VAR and LP literature \citep{Li:24,Olea:24}. Importantly, the backshifted VAR is easy to compute and does not require knowledge of the transition law of the state, which is also allowed to be endogenous. %

Third, I show in a nonparametric framework that state-dependent LPs estimate convex weighted averages of conditional marginal effects if the shock is independent of all other determinants of the outcome. This guarantee has been derived for linear LPs and generalized to state-dependent LPs with a binary state variable \citep{Rambachan:21,Kolesar:24}. My results also cover continuous and multidimensional states. These are common in applied work: of the 44 papers surveyed by \cite{Goncalves:24}, 19 use a continuous state variable. If the estimating equation \eqref{eq:state_dep_simple} does not accurately capture the relationship between state and effect, a projection of the effect on the interaction terms is recovered. A non-zero interaction coefficient therefore implies that the weighted average effect co-moves with the state, confirming the common interpretation as evidence of state-dependence. %

A nonparametric result is especially useful in the state-dependent case, which necessarily implies nonlinearities. There is no nonlinear parametric analog to the structural vector moving average model, which is general enough to encompass most linearized macroeconomic models. As a consequence, studies of state-dependent LPs that assume a parametric data generating process tend to apply to particular settings.\footnote{See the discussion of the related literature in this section for multiple examples.} My general result implies that the state-dependent LP estimand lies in the convex hull of all possible effects of the shock. Due to asymptotic equivalence, this is also true for the novel backshifted VAR. If there are no sign or size nonlinearities, both equal the average effect conditional on the state. This is a valuable minimal guarantee when the data generating process is unclear and the data do not allow for very flexible estimation.
    
In contrast, conventional state-dependent VAR estimands do not necessarily have a causal interpretation. They can lie outside the convex hull of all possible effects, even when there are no sign or size nonlinearities. This does not make the LP estimand more interesting for every data generating process.\footnote{An example in which estimating the VAR is preferable is given by \cite{Goncalves:24b}, who assume the data are generated by a Gaussian state-dependent VAR with endogenous state. In this case, the LP recovers the average effect of a marginal shock, while the estimated VAR together with a model of the state can be used to estimate, for example, the effect of a one standard deviation shock, which might be of greater interest for policy work.} But it may make LPs preferable in situations where the data generating process is uncertain. 

\paragraph*{Literature} This paper contributes to the literature relating LPs and VARs. \cite{Plagborg-Moller:21} first showed that both methods asymptotically yield the same impulse response estimands, so the choice between them should be made based on finite-sample considerations. This opened the door for studying the bias-variance tradeoff between the two methods \citep{Li:24,Olea:24,Xu:26}. Recently, \cite{Ludwig:24} derived a finite-sample equivalence between LPs and VARs. I show that the asymptotic equivalence breaks down in the state-dependent case and that the novel backshifted VAR restores it.

The nonparametric guarantee builds on cross-sectional work on linear regressions with a continuous regressor in a nonlinear environment, dating back at least to \cite{Yitzhaki:96} and \cite{Angrist:00}.\footnote{Closely related to this paper, \cite{Graham:22} show how to efficiently estimate averages of conditional projection coefficients and provide an interpretation as a weighted average derivative of the conditional expectation function.} \cite{Rambachan:21} first applied results of this literature to local projections and \cite{Kolesar:24} weakened the required regularity conditions. In a similar framework, \cite{Caravello:24} and \cite{Bousquet:26} show how to identify sign and size nonlinearities, \cite{Casini:25} focus on high-frequency event studies and \cite{Casini:26} show how to identify compliers in a time series setting. All these studies rely on a time series potential outcomes framework similar to the one pioneered by \cite{Angrist:11} and \cite{Angrist:18} for discrete treatments. In this setting, I derive the causal estimand of state-dependent LPs with continuous and multidimensional states, which covers virtually all specifications used in practice (see Supplemental Appendix S.1).
  
Multiple recent papers have studied state-dependent LPs in a parametric or semiparametric setting. \cite{Cloyne:23} extend the Kitagawa--Oaxaca--Blinder decomposition to decompose channels of impulse response heterogeneity. \cite{David:26} assume a DGP that is linear in the shock, which they argue is plausible for micro-to-macro settings. They develop a semiparametric estimator to flexibly model the dependence of dynamic effects on a state. In a similar context, \cite{Almuzara:26} assume that the data are generated by a structural vector moving average model with coefficients varying across units. They characterize the estimand of state-dependent panel LPs and make recommendations for inference. \cite{Manuel:26} use policy rule changes to identify policy effects with state-dependent LPs. They leverage partial linearity of the private sector block for their identification result. \cite{You:26} studies how well state-dependent LPs capture nonlinearities if the data are generated by a quadratic VAR. Assuming the data are generated by a state-dependent VAR, \cite{Goncalves:24} show that, with Gaussian shocks, state-dependent LPs estimate the effect of a marginal shock, which might not be of interest. In follow-up work, they develop semi- and nonparametric methods to more precisely estimate the effect of interest \citep{Goncalves:24b,Goncalves:26}. My result complements this literature by clarifying what state-dependent LPs estimate when the data generating process is unknown.
     
\paragraph*{Outline} Section \ref{sec:runex} previews the main results of this paper, setting up a simple running example. Section \ref{sec:ssvar} discusses the non-equivalence of LPs and VARs and introduces the new estimator. Section \ref{sec:observed} derives the causal estimand of state-dependent LPs. Section \ref{sec:sim} compares different state-dependent VAR estimators to state-dependent LPs in an empirically calibrated simulation study, while Section \ref{sec:empirics} applies the estimators in an empirical setting, revisiting the \cite{Ramey:18} study on state-dependent effects of military spending news.

\section{Running Example and Preview of Main Results}\label{sec:runex}

Researchers who want to estimate the response of an outcome $Y_{t+h}$ to a change in the shock $X_t$ conditional on an observed state variable have two standard tools: state-dependent local projections (LPs) and state-dependent vector autoregressions (VARs). In the unconditional case, both methods are known to be equivalent asymptotically \citep{Plagborg-Moller:21}. The following running example, used throughout the paper, shows that this equivalence breaks down in the state-dependent case and that the impulse response as it is usually computed from the VAR does not have to correspond to any meaningful causal effect summary.

\begin{example}
  Consider the following Markov-switching autoregressive model:
  \begin{equation}\label{eq:markov}
      Y_t = \rho(S_{t-1})Y_{t-1} + a(S_t) X_t.
  \end{equation}
   Let the state be binary, $S_{t} \in \{0,1\}$, and Markov with known probabilities of staying, $\pi_{00}$ and $\pi_{11}$. The shock $X_t$ is i.i.d. with mean zero and independent of past outcomes $Y_{t-k},~k >0$ and all leads and lags of the state $S_{t+k},~k \in \mathbb{Z}$. The outcome $Y_t$ is covariance stationary with mean zero.
   
   Fix a path of states $(S_t, \ldots, S_{t+h-1})$. Along that path, iterating on \eqref{eq:markov} gives the response of $Y_{t+h}$ to $X_t$, also called the dynamic causal effect of $X_t$: 
   \begin{equation}\label{eq:posseff}
      a(S_t) \rho(S_t) \cdot \ldots \cdot \rho(S_{t+h-1}),
   \end{equation}
   which is the shock impact $a(S_t)$ times its autoregressive propagation. Since there are many possible paths of states, there are also many possible effects. A minimum requirement of any useful summary is that it is in the convex hull of all effects \eqref{eq:posseff}. A natural summary is the average effect of $X_t$ on $Y_{t+h}$, conditional on some initial state $S_{t-1}=s$, which is given by 
   \begin{equation}\label{eq:aveffect}
      \E_s \left[ a(S_t) \rho(S_t) \cdot \ldots \cdot \rho(S_{t+h-1}) \right],
   \end{equation}
   where $\E_s[\bullet] \equiv \E[\bullet \mid S_{t-1}=s]$. 
     
   Suppose a researcher observes the time series $\{Y_t,S_{t-1},X_t\}_{t=0}^T$ and knows the transition law of $S_t$, but not the data generating process \eqref{eq:markov}. The effect of $X_t$ on $Y_{t+h}$ conditional on the initial state $S_{t-1}=s$ can be estimated using LPs. Under standard regularity conditions, the coefficients of the LP \eqref{eq:state_dep_simple} have probability limits that satisfy
     \begin{equation*}
       \beta_0^h + s \beta_1^h = \frac{\E_s[Y_{t+h}X_t]}{\E_s[X_t^2]},
     \end{equation*}
     and a direct calculation shows that the LP recovers the average effect \eqref{eq:aveffect} exactly.

     As an alternative to LPs, researchers can also estimate a state-dependent VAR. This amounts to regressing $Y_t$ and $X_t$ on their own lags conditional on $S_{t-1}$. Using the well-known Cholesky identification scheme with the shock $X_t$ ordered first, the state-dependent VAR has the following probability limit for the example DGP:

     \begin{equation}\label{eq:runex-var}
         \begin{pmatrix}
         X_t \\ Y_t
         \end{pmatrix} = \begin{pmatrix}
          0 & 0 \\ 0 & \rho(S_{t-1})
         \end{pmatrix} \begin{pmatrix} 
          X_{t-1} \\ Y_{t-1}
         \end{pmatrix} + \begin{pmatrix}
          1 & 0 \\ \E_{S_{t-1}}[a(S_t)] & 1
         \end{pmatrix} \begin{pmatrix}
             X_t \\ \left(a(S_t)- \E_{S_{t-1}}[a(S_t)]\right)X_t
         \end{pmatrix}.
     \end{equation}

    \begin{figure}[t]
      \centering
      \caption{State-Dependent VAR and LP Conditional on $S_{t-1}=0$ (Running Example).}
      \label{fig:no_girf}
      \begin{tikzpicture}
      \begin{axis}[
          width=12.2cm, height=8.4cm,
          xlabel={Horizon $h$}, xlabel style={font=\small},
          xmin=-0.2, xmax=6.3, ymin=0, ymax=0.9,
          xtick={0,1,2,3,4,5,6}, ytick={0.2,0.4,0.6,0.8},
          axis lines=left, tick style={black},
          x tick label style={font=\small}, y tick label style={font=\small},
          clip=true,
          legend style={at={(0.97,0.97)}, anchor=north east,
              draw=black!25, fill=white, font=\small,
              row sep=2pt, /tikz/every even column/.append style={column sep=6pt}},
          legend cell align=left,
      ]
      \addplot[name path=lo, draw=none, forget plot] coordinates
        {(0,0.5000) (1,0.3000) (2,0.0600) (3,0.0120) (4,0.0024) (5,0.0005) (6,0.0001)};
      \addplot[name path=hi, draw=none, forget plot] coordinates
        {(0,1.5000) (1,0.4250) (2,0.3612) (3,0.3071) (4,0.2610) (5,0.2219) (6,0.1886)};
      \addplot[black!9, forget plot] fill between[of=lo and hi];
      \addplot[draw=black!35, line width=0.5pt, forget plot] coordinates
        {(0,1.5000) (1,0.4250) (2,0.3612) (3,0.3071) (4,0.2610) (5,0.2219) (6,0.1886)};
      \addplot[draw=black!35, line width=0.5pt, forget plot] coordinates
        {(0,0.5000) (1,0.3000) (2,0.0600) (3,0.0120) (4,0.0024) (5,0.0005) (6,0.0001)};
      \addlegendimage{figred, line width=1.2pt, dashed, mark=triangle*, mark size=3pt}
      \addlegendentry{VAR (Fixed State)}
      \addlegendimage{figred, line width=1.2pt, mark=square*, mark size=2.7pt}
      \addlegendentry{VAR (Moving State)}
      \addlegendimage{black, line width=1.6pt, mark=*, mark size=2.7pt}
      \addlegendentry{True Conditional Effect / LP}
      \addlegendimage{area legend, fill=black!9, draw=black!35}
      \addlegendentry{Feasible Shock Effects}
      \addplot[figred, line width=1.2pt, dashed, mark=triangle*, mark size=3pt, forget plot] coordinates
        {(0,0.7000) (1,0.5950) (2,0.5057) (3,0.4299) (4,0.3654) (5,0.3106) (6,0.2640)};
      \addplot[figred, line width=1.2pt, mark=square*, mark size=2.7pt, forget plot] coordinates
        {(0,0.7000) (1,0.5040) (2,0.3520) (3,0.2442) (4,0.1693) (5,0.1173) (6,0.0812)};
      \addplot[black, line width=1.6pt, mark=*, mark size=2.7pt, forget plot] coordinates
        {(0,0.7000) (1,0.4000) (2,0.2646) (3,0.1815) (4,0.1254) (5,0.0869) (6,0.0602)};
      \end{axis}
      \end{tikzpicture}
      
      \notes{The shaded band is the convex hull of possible shock effects, conditional on any path of states $S_{t},\ldots,S_{t+h-1}$. The black line is the true conditional effect, equal to the state-dependent LP estimand. The dashed dark red line is the fixed-state
    VAR estimand, the solid dark red line is the moving-state VAR estimand. Both quantities are described in the example. Calibration: $\rho(1)=0.20$, $a(1)=1.5$;
      $\rho(0)=0.85$, $a(0)=0.5$; $\pi_{00}=\pi_{11}=0.80$.}
      \end{figure}

      The effect of $X_t$ on $Y_{t+h}$ conditional on $S_{t-1}=s$ is usually computed by iterating on the autoregression, holding future states fixed at $S_{t-1}=s$ \citep[e.g.][]{Auerbach:12}. This yields the effect estimand $\rho(s)^h \E_s[a(S_t)]$. \cite{Auerbach:12,Auerbach:13} are aware that this estimand does not account for the possibility that the economy might move out of state $s$ between $t-1$ and $t+h-1$. Thus, the natural benchmark for this fixed-state estimator is the effect of $X_t$ on $Y_{t+h}$ conditional on $S_{t-1} = \ldots = S_{t+h-1} = s$. Figure \ref{fig:no_girf} plots the fixed-state VAR estimand (dashed line) together with the responses of $Y_{t+h}$ to $X_t$ along all possible state paths (gray band) and the average effect (black line). The fixed-state estimand is not only different from this expected response, but it also lies outside the range of all feasible shock effects at every plotted horizon $h \geq 1$. It does so because it pairs the average impact $\E_s[a(S_t)]$ with the own-state persistence $\rho(s)^h$, a combination no single path of state realizations delivers for the plotted calibration. In particular, it is not even the response obtained by holding the regime fixed at
      $s$ over the whole horizon, which would be $a(s)\rho(s)^h$.

      A VAR estimand that is more comparable to the LP estimand would average over all possible transitions between states. This is typically infeasible, unless the researcher takes a strong stance on the transition law of $S_t$. In this thought experiment the law is known, so the `moving-state' estimator can be computed and equals $\E_s[\rho(S_{t+h-1}) \cdot \ldots \cdot \rho(S_t)]\E_s[a(S_t)]$. It averages propagation and impact separately and then combines them, while the effect \eqref{eq:aveffect} combines them along each path and averages afterwards. The two differ by the covariance term $\operatorname{Cov}_s[\rho(S_{t+h-1})\cdot \ldots \cdot\rho(S_t), a(S_t)]$, which neither VAR estimator can capture. As a consequence, the moving-state estimand is generally different from the average effect and can itself lie outside the set of feasible shock effects (solid dark red line in Figure \ref{fig:no_girf}).
\end{example}

In the running example, the LP recovers the correct effect, while the VAR as it is usually computed targets a different object, which need not be an average of causal effects at all---it can fall outside the convex hull of all possible effects. This gap is not the truncation bias of unconditional VAR analysis. In the example, the error of the one-lag state-dependent projection, $(a(S_t) - \E_{S_{t-1}}[a(S_t)])X_t$, is conditionally uncorrelated with the entire past, so adding further lags would not change the VAR. The gap is due only to the co-movement between impact and propagation. 
  
The running example will turn out to be a special case of the results in Sections~\ref{sec:ssvar} and~\ref{sec:observed}. Section \ref{sec:ssvar} shows why exactly the equivalence between LPs and VARs \citep{Plagborg-Moller:21,Ludwig:24} breaks down in the state-dependent case, and constructs a VAR-based estimand whose probability limit is again the LP estimand. Section \ref{sec:observed} then shows that this estimand has a causal interpretation under very general conditions.

\section{State-Dependent Local Projections and Vector Autoregressions}\label{sec:ssvar}

This section explains why state-dependent VARs and LPs disagree asymptotically and constructs a VAR-based estimand that makes them agree. %

\subsection{The Conditional Projection Decomposition}\label{subsec:projection_model}

Begin by stacking the shock $X_t$ and the outcome $Y_t$ in a vector 

\begin{equation*}
    \boldsymbol{Y}_t = \begin{pmatrix}
        X_t \\ Y_t
    \end{pmatrix}.
\end{equation*}
It simplifies the analysis to assume that the shock is independent of the past, letting $\ind$ denote statistical independence: %
  
\begin{assumpexoVAR}\label{ass:exo_var}
    For all $t$, $X_{t} \ind (\boldsymbol{Y}_{t-1}',S_{t-1},\boldsymbol{Y}_{t-2}',S_{t-2},...)$ and $\E[X_t]=0$.
\end{assumpexoVAR}
  
Next, define $P_s[\bullet|\bullet]$ as the projection operator with respect to the conditional expectation $\E_s[\bullet]$, and let $P_{\mathcal{S}}[\bullet|\bullet]$ be the projection with respect to the unconditional expectation $\E[\bullet]$. Throughout the section, $S_{t-1} \in \{0,1\}$ is assumed. With a binary state, the coefficients of the state-dependent LP
\begin{equation*}
    Y_{t+h} = X_t \beta_0^h + S_{t-1} X_t \beta_1^h + \text{error}_{h,t+h}
\end{equation*}
satisfy 
\begin{equation*}
    P_s[Y_{t+h} \mid X_t] = (\beta_0^h + s \beta_1^h) X_t.
\end{equation*}

Now the reduced-form VAR conditional projection with one lag can be defined via 
\begin{align*}
   \boldsymbol{Y}_t & = \boldsymbol{\Pi}(S_{t-1}) \boldsymbol{Y}_{t-1} + \boldsymbol{E}_t , \\
   \boldsymbol{\Pi}(S_{t-1}) \boldsymbol{Y}_{t-1}& \equiv P_{S_{t-1}}[\boldsymbol{Y}_t \mid \boldsymbol{Y}_{t-1}],
\end{align*}
where $\E_{S_{t-1}}[\boldsymbol{E}_t \boldsymbol{Y}_{t-1}'] = \mathbf{0}$. This can be generalized to a projection with $p$ lags by adding own lags to $\boldsymbol{Y}_t$ (the companion form). By applying the common recursive identification scheme and utilizing that $X_t$ is independent of the past, there is a structural VAR (SVAR) representation of $\boldsymbol{Y}_t$ in terms of projection coefficients:
\begin{equation}\label{eq:SVAR_form}
    \boldsymbol{Y}_t = \boldsymbol{\Pi}(S_{t-1})\boldsymbol{Y}_{t-1} + \boldsymbol{A}(S_{t-1}) \begin{pmatrix}
        X_t \\ e_t^\perp
    \end{pmatrix},
\end{equation}  
where $\boldsymbol{A}(S_{t-1})$ is lower triangular and $\E[X_t e_t^\perp \mid S_{t-1}]=0$.\footnote{Formally, denote the elements of the reduced-form error as $(X_t,e_t)' = \boldsymbol{E}_t$. Then the orthogonalized projection error $e_t^\perp$ is defined via
\begin{align*}
      e_t^\perp & \equiv e_t - P_{S_{t-1}}[e_t \mid X_t].
\end{align*}
Lastly, the contemporaneous slope coefficients are computed as 
\begin{equation*}
    \boldsymbol{A}(s) \equiv \operatorname{chol}(\E[\boldsymbol{E}_t \boldsymbol{E}_t' \mid S_{t-1}=s]) \times \operatorname{diag}(\E[X_t^2], \E[(e_t^\perp)^2 \mid S_{t-1}=s])^{-1/2},
\end{equation*}
where $\operatorname{chol}$ denotes the Cholesky decomposition. In companion form, the impact matrix in \eqref{eq:SVAR_form} stacks this matrix on top of zeros.
    }
Despite looking like a structural model, this representation is defined purely in terms of population moments and exists under minimal regularity conditions. The only economic assumption so far is $X_t$ being independent of the past. The orthogonalized error $e_t^\perp$, however, need not be independent of $X_t$ or independent over time.
  
The conditional projection coefficients can be used to construct impulse response estimands in an iterative way. This has already been introduced for the example in Section \ref{sec:runex} and is now repeated more generally. The most straightforward effect estimand is
\begin{equation*}
    \theta^f_{\mathit{VAR,h}}(s) \equiv \left( \boldsymbol{\Pi}(s)^h \boldsymbol{A}(s) \right)_{2,1},
\end{equation*}
where $f$ stands for fixed state. This is the impulse response estimand used by \cite{Auerbach:12}. They are aware that this estimand does not account for the possibility that the economy might move out of state $s$ between time $t-1$ and $t+h-1$. Since it is well known that LPs average over future state changes, it is no surprise that $\theta_{\mathit{VAR,h}}^f$ will be different from the LP estimand. Perhaps more surprisingly, the running example of Section \ref{sec:runex} shows that this estimand does not have to correspond to any causal quantity. 
  
An effect estimand that accounts for the possibility of future state changes would be 
\begin{equation*}
    \theta^m_{\mathit{VAR,h}}(s) \equiv \left( \E_s[\boldsymbol{\Pi}(S_{t+h-1}) \cdot \ldots \cdot \boldsymbol{\Pi}(S_t) ] \boldsymbol{A}(s) \right)_{2,1},
\end{equation*}
where $m$ stands for moving state. As derived by \cite{Goncalves:24}, in a state-dependent VAR model with fully exogenous state and independent error terms this is the response of $Y_{t+h}$ to a shock $X_t$ of arbitrary size.\footnote{See Proposition 3.1 of \cite{Goncalves:24}. For this data generating process, $\theta^m_{\mathit{VAR,h}}(s)$ is both what they call the conditional average response and the conditional marginal response.} Since this estimand averages over future paths of  the state, it is a natural comparison to the LP estimand. However, in practice it cannot be computed without taking a strong stance on the evolution of the state. This estimand rather serves as an infeasible benchmark that separates the known discrepancy between state-dependent LP and VAR estimands, $\theta_{\mathit{VAR,h}}^f(s) - \theta_{\mathit{VAR,h}}^m(s)$, from the novel point raised in this study, which is the difference between $\theta_{\mathit{VAR,h}}^m(s)$ and the LP estimand.

\subsection{Understanding LP--VAR Non-Equivalence}\label{subsec:non-equivalence}

Section \ref{sec:runex} showed that the VAR-based estimands $\theta_{\mathit{VAR,h}}^f$ and $\theta_{\mathit{VAR,h}}^m$ generally disagree with the LP estimand. This contrasts with the equivalence results of \cite{Plagborg-Moller:21} and \cite{Ludwig:24} for linear LPs and VARs. To understand why this equivalence breaks in the state-dependent case, it is necessary to revisit why it holds in the linear case.
  
To make the main point, it is sufficient to focus on the effect at horizon $h=1$. %
At this horizon, the unconditional LP estimand corresponds to the coefficient of the projection $P_\mathcal{S}[Y_{t+1} \mid X_t]$. This can be re-written as the sum of a projection on the prediction using time $t$ information and a projection on the prediction error:
\begin{align*}
  P_\mathcal{S}[Y_{t+1} \mid X_t]  & = P_{\mathcal{S}} \left[ P_\mathcal{S}[Y_{t+1} \mid \boldsymbol{Y}_t] + \left( Y_{t+1} - P_\mathcal{S}[Y_{t+1} \mid \boldsymbol{Y}_t] \right) \mid X_t \right] \\
       & = \underbrace{P_{\mathcal{S}} \left[ \boldsymbol{\Pi}_{2,\bullet} \boldsymbol{Y}_{t} \mid X_t \right]}_{= \left(\boldsymbol{\Pi} \boldsymbol{A} \right)_{2,1} X_t} + \underbrace{P_\mathcal{S} \left[Y_{t+1} - P_\mathcal{S}[Y_{t+1} \mid \boldsymbol{Y}_t] \mid X_t \right]}_{=0}.
\end{align*}

Due to the equalities under the brackets, the LP estimand is equal to the VAR estimand $\left( \boldsymbol{\Pi} \boldsymbol{A} \right)_{2,1}$. This holds for two reasons: First, projecting the forecast $P_\mathcal{S}[Y_{t+1} \mid \boldsymbol{Y}_t]$ onto $X_t$ amounts to multiplying the one-step projection coefficient $\boldsymbol{\Pi}$ with the projection $P_{\mathcal{S}}[\boldsymbol{Y}_{t} \mid X_t] = \boldsymbol{A}_{\bullet,1} X_t$. Due to the linearity of the projection operator, the constant matrix $\boldsymbol{\Pi}$ can be pulled out and $\boldsymbol{Y}_t$ can be predicted separately.
    
Second, the prediction error $Y_{t+1} - P_\mathcal{S}[Y_{t+1} \mid \boldsymbol{Y}_t]$ cannot be predicted using $X_t$, which is already included in $\boldsymbol{Y}_t$. Thus, predicting future predictions leads to no information loss compared to predicting the future value itself. This is the tower property of linear projections, also called the law of iterated projections. The equivalence results of \cite{Plagborg-Moller:21} and \cite{Ludwig:24} effectively apply this result.
    
In the state-dependent case, it turns out that both properties fail: conditional projection coefficients \textit{cannot} be predicted separately from outcomes and future prediction errors \textit{can} generally be predicted using $X_t$, conditional on $S_{t-1}$. To make this point more concrete, write the same decomposition as above for the state-dependent case:
\begin{equation}\label{eq:LP-Decomp}
    P_s[Y_{t+1} \mid X_t] = \underbrace{P_s \left[ \boldsymbol{\Pi}(S_t)_{2,\bullet} \boldsymbol{Y}_t \mid X_t\right]}_{\neq \E_s[\boldsymbol{\Pi}(S_t)]_{2,\bullet}  P_s[\boldsymbol{Y}_{t}\mid X_t]} + \underbrace{P_s \left[Y_{t+1} - P_{S_t}[Y_{t+1} \mid \boldsymbol{Y}_t] \mid X_t \right]}_{\neq 0}.
\end{equation}

The VAR-based estimand $\theta_{\mathit{VAR,1}}^m(s)$ would agree with the state-dependent LP estimand if the two inequalities under the brackets were equalities. To understand why this is generally not the case, it is again instructive to consider the running example.

\begin{example}
   Recall model \eqref{eq:markov} and its state-dependent VAR representation \eqref{eq:runex-var}. For this model, equation \eqref{eq:LP-Decomp} has the form 
   \begin{align*}
        P_s[Y_{t+1} \mid X_t] = \underbrace{\E_s[\rho(S_t) a(S_t)]}_{\neq \E_s[\rho(S_t)] \E_s[a(S_t)]} X_t + \underbrace{P_s[\left(a(S_{t+1}) - \E_{S_t}[a(S_{t+1})]  \right) X_{t+1} \mid X_t]}_{=0}.
   \end{align*}
   
   The reason the VAR-based estimand $\theta_{\mathit{VAR,1}}^m(s)$ and the LP estimand do not agree in this example is the missing covariance term $\operatorname{Cov}_s[\rho(S_t),a(S_t)]$. The autoregressive and the impact component cannot be predicted separately, since this separation ignores their co-movement. Equation \eqref{eq:LP-Decomp} shows that this missing covariance is not specific to the running example but one of two main reasons state-dependent LPs and VARs disagree.
       
   The running example lacks one of the two ways state-dependent LP--VAR equivalence can break down: prediction errors are multiples of the independent $X_t$, so they cannot be predicted ahead of time. To see when this can fail, augment \eqref{eq:markov} with a state-dependent moving-average coefficient:
   \begin{equation*}
      Y_t = \rho Y_{t-1} + a X_t + c(S_{t-2})X_{t-1},
   \end{equation*}
   where $\rho$ and $a$ are now held fixed across states to isolate the second reason for disagreement. The state-dependent VAR form becomes 
   \begin{equation*}
    \begin{pmatrix}
        X_t \\ Y_t
    \end{pmatrix} = \begin{pmatrix}  0 & 0 \\
    \E_{S_{t-1}}[c(S_{t-2})] & \rho
  \end{pmatrix} 
  \begin{pmatrix} X_{t-1} \\ Y_{t-1}
  \end{pmatrix} + 
  \begin{pmatrix} 1 & 0 \\
      a & 1
  \end{pmatrix} 
  \begin{pmatrix}
      X_t \\ \left(c(S_{t-2}) - \E_{S_{t-1}}[c(S_{t-2})] \right)X_{t-1}
  \end{pmatrix}.
   \end{equation*}
  The decomposition \eqref{eq:LP-Decomp} is now
  \begin{equation*}
    P_s[Y_{t+1} \mid X_t] = \underbrace{\left(\rho a + \E_s[\E_{S_t}[c(S_{t-1})]] \right)}_{= \theta_{\mathit{VAR,1}}^m(s)} X_t + \underbrace{\left(c(s) - \E_s \left[\E_{S_t}[c(S_{t-1})] \right] \right)}_{\neq 0} X_t,
  \end{equation*}
  where the second term is non-zero since $\E_s \left[\E_{S_t}[c(S_{t-1})] \right]$ sits strictly between $c(0)$ and $c(1)$ if the switching probabilities of the states are all strictly positive. The LP--VAR equivalence fails since future prediction errors can be predicted conditional on $S_{t-1}$.
\end{example}

In summary, state-dependent VARs fail to match LPs since they do not hold the projection operator fixed for their iterative predictions but let it fluctuate randomly, conditioning on $S_{t-1}$ to predict $t$, on $S_t$ to predict $t+1$, and so on. In contrast, LPs use a single projection conditional on $S_{t-1}$ to predict the outcome $Y_{t+h}$. Panels (a) and (b) of Figure \ref{fig:lp_var} illustrate this difference. This fluctuation in projection operators is responsible for random projection coefficients $\boldsymbol{\Pi}(S_{t+k-1})$ that can co-move with the variables in the VAR and for prediction errors that can be forecasted using past shocks. By \eqref{eq:LP-Decomp}, these two phenomena are responsible for state-dependent LPs and VARs disagreeing. 
  
Next, I propose a VAR-based estimator that holds the projection operator fixed and argue that it is the natural iterative analog of state-dependent LPs.

\subsection{The Backshifted VAR Estimator}\label{subsec:alternative_est}

So far I have shown that state-dependent VARs depart from state-dependent LPs
in two ways. They replace the direct projection with an iterated one, and they
switch the projection operator from one step to the next. Only the first
separates LPs from VARs in the unconditional case; the second is an artifact of
state-dependence, and the source of practical difficulties. Accounting for future
state changes requires modeling the transition law of $S_t$, and even when this
is done correctly the resulting estimand $\theta^m_{\mathit{VAR,h}}$ can fail to
have a causal interpretation, as the running example of Section
\ref{sec:runex} shows.

To remove this artifact while keeping the VAR structure, I introduce a new VAR-based estimator that holds the projection operator fixed across steps, applying $P_{S_{t-1}}$ at every horizon. The
resulting estimator then differs from the state-dependent LP in only one way, by
iterating one-step projections rather than projecting directly, which is the
difference that separates unconditional VARs and LPs. Its probability limit is therefore the state-dependent LP estimand.

To develop the new estimator formally, consider $h+1$ state-dependent VAR projections where the state is shifted back lag after lag:
\begin{align}
    \boldsymbol{Y}_t & = \boldsymbol{\Pi}^0(S_{t-1}) \boldsymbol{Y}_{t-1} + \boldsymbol{A}^0(S_{t-1}) \boldsymbol{E}_t^{0,\perp}, \notag \\
    \vdots & ~~~~~~~~~~~~~~~~~~~~~~~~~~~ \vdots \notag \\
    \boldsymbol{Y}_t & = \boldsymbol{\Pi}^h(S_{t-1-h}) \boldsymbol{Y}_{t-1} + \boldsymbol{A}^h(S_{t-1-h}) \boldsymbol{E}_t^{h,\perp}. \label{eq:vars_decreasing}
\end{align}
As before, the projections use $p$ lags and are written in companion form. Under Assumption \ref{ass:exo_var}, the orthogonalized projection error is of the form $\boldsymbol{E}_t^{k,\perp} = (X_t, e_t^{k,\perp})'$. These projections are just as described in \eqref{eq:SVAR_form} with the difference that for the $k$th projection the conditional expectation $\E[\bullet | S_{t-1}=s]$ is replaced with $\E[\bullet| S_{t-1-k} = s]$. Iterating forward, using the $k$th projection for the $k$th prediction step yields the impulse response estimand
\begin{equation}\label{eq:theta_b}
    \theta^b_{\mathit{VAR,h}}(s) \equiv (\boldsymbol{\Pi}^h(s) \cdot \ldots \cdot \boldsymbol{\Pi}^1(s) \boldsymbol{A}^0(s))_{2,1},
\end{equation}
where $b$ stands for backshifted state.\footnote{This iterative combination of multiple different VAR projections is similar in spirit to \citeauthor{Ludwig:24}'s (\citeyear{Ludwig:24}) VAR-sequence. Using this technique, he is able to prove a finite-sample equivalence between linear VARs and LPs. However, he combines linear VAR projections with different lag lengths, while I combine state-dependent VAR projections that condition on different lags of the states.} It turns out that $\theta^b_{\mathit{VAR,h}}(s)$ is identical to the state-dependent LP estimand for all horizons up to the lag order $p$. 

  \begin{prop}\label{prop:var_lp}
    Assume $(\boldsymbol{Y}_t',S_{t-1})$ is covariance stationary, $X_t$ has strictly positive variance and Assumption \ref{ass:exo_var} holds. Then, for all $h \leq p$,
    \begin{equation*}
        \theta^b_{\mathit{VAR,h}}(s) = \frac{\E[Y_{t+h}X_t \mid S_{t-1}=s]}{\E[X_t^2 \mid S_{t-1} =s]}.
    \end{equation*}
    For $h > p$, the equality continues to hold if $\E[\boldsymbol{E}_t^{k,\perp} X_{t-k} \mid S_{t-1-k} = s] = \mathbf{0}$ for $k = p+1, \ldots, h$.
\end{prop}
The proof is in Appendix~\ref{proof:recursive}.

The first part is the state-dependent analog of the finite-lag equivalence of \cite{Plagborg-Moller:21}: with $p$ lags, LPs and VARs agree up to horizon $p$. Beyond horizon $p$, the backshifted VAR can carry truncation bias, just like an unconditional VAR($p$). The condition in the second part rules this out.

  To understand why the equivalence holds, it is useful to again recall the decomposition \eqref{eq:LP-Decomp} for the state-dependent LP estimand at $h=1$. If the prediction of $Y_{t+1}$ given the past $\boldsymbol{Y}_t$ conditions on $S_{t-1}=s$ instead of $S_t$, the decomposition reads 
  
  \begin{equation*}
    P_s[Y_{t+1} \mid X_t] = \underbrace{P_s \left[ \boldsymbol{\Pi}^1(s)_{2,\bullet} \boldsymbol{Y}_t \mid X_t\right]}_{= \boldsymbol{\Pi}^1(s)_{2,\bullet}  P_s[\boldsymbol{Y}_{t}\mid X_t]} + \underbrace{P_s \left[Y_{t+1} - P_{s}[Y_{t+1} \mid \boldsymbol{Y}_t] \mid X_t \right]}_{= 0},
\end{equation*}
  so the state-dependent LP estimand is identical to $\theta^b_{\mathit{VAR,1}}(s)$. The reason is that the projection operator is held fixed, so the autoregressive coefficient can be pulled out of the projection and future prediction errors cannot be predicted. The tower property of projections holds. Figure \ref{fig:lp_var} (c) visualizes the prediction underlying state-dependent LPs, conventional state-dependent VARs and the new backshifted VAR estimand.

  The backshifted estimand $\theta^b_{\mathit{VAR,h}}$ is useful as a theoretical device to understand how LPs and VARs relate in the state-dependent case. But it also has practical appeal: it is easy to compute and does not require knowledge of the transition law of $S_t$. Moreover, it matches the LP estimand and therefore shares its causal interpretation, which Section~\ref{sec:observed} now establishes.

\section{Causal Interpretation of State-Dependent Local Projections}\label{sec:observed}

This section sets up a nonparametric econometric framework and reviews an important result of \citet{Rambachan:21} and \citet{Kolesar:24} on the causal interpretation of linear LPs in the presence of nonlinearities. It then extends this result to state-dependent LPs. Notation follows \citet{Kolesar:24}. 

\subsection{Econometric Framework} \label{sec:framework}

Let $Y_{t+h}$ be a scalar outcome variable and $X_t$ be a macroeconomic shock. As is common in the applied literature, I assume that the shock $X_t$ is observed without measurement error, which makes a regression of $Y_{t+h}$ on $X_t$ feasible.\footnote{With classical measurement error, attenuation bias will yield a rescaled version of this regression, which leaves the shape of the estimated impulse response intact \citep{Plagborg-Moller:21}.} Without loss of generality, let $Y_{t+h}$ be determined by the \textit{structural function}
\begin{equation}\label{eq:str_func}
  Y_{t+h} = \psi_h(X_t, U_{h,t+h}),
\end{equation}
where $U_{h,t+h}$ is a collection of variables that influence the outcome variable. In most macroeconomic models, $U_{h,t+h}$ would be a collection of shocks, lags of $Y_t$ and other macroeconomic variables that affect $Y_{t+h}$. %

\begin{example}
    In the running example, the structural function $\psi_h(X_t,U_{h,t+h})$ can be calculated by iterating on $\psi_0(X_t,U_{0,t}) = \rho(S_{t-1})Y_{t-1} + a(S_t)X_t$. The collection $U_{h,t+h}$ is then $\{S_{t-1}, \ldots, S_{t+h}, Y_{t-1},X_{t+1}, \ldots, X_{t+h} \}$. %
\end{example}

In practice, the econometrician may know neither the functional form of $\psi_h$ nor the variables $U_{h,t+h}$ nor the transition law of the state $S_t$. It will prove useful to marginalize the structural function over $U_{h,t+h}$ conditional on a lagged state $S_{t-1} \in \mathcal{S}$ to obtain the \textit{conditional average structural function} 
\begin{equation*}
    \Psi_h(x,s) \equiv \E_s[\psi_h(x,U_{h,t+h})],
\end{equation*}
where again the expectation $\E_s[\bullet] \equiv \E[\bullet \mid S_{t-1}=s]$ is conditional on the lagged state. With the convention $\E_{\mathcal{S}}[ \bullet] = \E[\bullet]$, this definition nests the average structural function of \cite{Blundell:03}:
\begin{equation*}
     \Psi_h(x,\mathcal{S}) = \E[\psi_h(x,U_{h,t+h})].
\end{equation*}

The structural function can now be used to define dynamic causal effects. In nonlinear time series models, the response of $Y_{t+h}$ to a change from $X_t$ to $X_t + \delta$ depends on the history of past shocks, the baseline shock level $X_t$ and the sign as well as absolute size of $\delta$. Therefore, there are many different effects one could possibly consider. I focus on \textit{average marginal effects} of the form 
\begin{equation}\label{eq:av_marg_eff}
    \theta_h(s,\omega) \equiv \int \omega(x) \Psi_h'(x,s)dx,
\end{equation}
where the weight function $\omega$ is non-negative and satisfies $\int\omega(x)dx = 1$. As noted by \cite{Kolesar:24}, this class contains many interesting effect summaries, in particular the expected marginal effect $\E[\Psi_h'(X_t,s)]$ and the average effect of a shock of size $\delta>0$, $\frac{1}{\delta}\E[\Psi_h(X_t+\delta,s)-\Psi_h(X_t,s)]$.\footnote{The first quantity corresponds to $\omega(x) = f_X(x)$ and the second to $\omega(x) = \frac{1}{\delta}\int_{x-\delta}^x f_X(z)dz$, where $f_X(x)$ is the density of the shock. An effect summary that cannot be written as an average marginal effect is, for example, the response of $Y_{t+h}$ to $X_t$ conditional on the state being held fixed, $S_{t+h-1}=\ldots = S_{t-1} = s$.}

\subsection{The Causal Estimand of State-Dependent Local Projections}\label{sec:main}

Having set up the machinery, I now consider state-dependent local projections of the form 
\begin{equation}\label{eq:state_dep_lp}
  Y_{t+h} =  f(S_{t-1})' X_t \beta^h + \text{error}_{h,t+h},
  \end{equation}
  where the data are generated by the structural function \eqref{eq:str_func}, the lagged state $S_{t-1}$ takes on values in $\mathcal{S}$, $f: \mathcal{S} \to \mathbb{R}^k$ is a function mapping states to interaction terms and $\beta^h \in \mathbb{R}^k$ is the regression coefficient. For example, in \cite{Ramey:18}, $S_{t-1}$ is the unemployment rate and $f$ consists of two indicator functions defining a slack and an expansionary state, respectively:
  \begin{equation*}
  f(s) = \begin{pmatrix}
        \I[s \geq 6.5\%] \\ 1-\I[s \geq 6.5\%]
  \end{pmatrix}.
  \end{equation*}

  To provide a causal interpretation for state-dependent LPs, I build on an important identity popularized by \cite{Yitzhaki:96} and \cite{Angrist:00}, which \cite{Rambachan:21} and \cite{Kolesar:24} applied to linear local projections. Given the independence $X_t \ind U_{h,t+h}$ and additional regularity conditions spelled out in Appendix~\ref{app:causal}, they showed that linear LPs estimate an average marginal effect:
  \begin{equation}\label{eq:lp_lin_effect}
      \frac{\E \left[ Y_{t+h}X_t \right]}{\E \left[ X_t^2 \right]} = \theta_h(\mathcal{S},\omega_X) = \int \omega_X(x) \Psi_h'(x,\mathcal{S}) dx,
  \end{equation}
  where the weight function
  \begin{equation}\label{eq:omega-x}
      \omega_X(x) \equiv \frac{\operatorname{Cov}\left[\I[X_t \geq x],X_t \right]}{\V[X_t]}
  \end{equation}
  is non-negative, integrates to one and peaks at zero. Since $\omega_X$ is a function of the marginal shock distribution only, the weights corresponding to a particular shock series can be estimated and are identical across applications. 
    
  It turns out that state-dependent LPs have a causal estimand under effectively the same conditions as in the linear case. To ease notation, from now on let $f_{t-1}$ denote $f(S_{t-1})$.

    \begin{assumpexo}\label{ass:exo2}
    For all $h \geq 0$, $t \in \mathbb{Z}$, $X_t \ind (S_{t-1},U_{h,t+h})$ and $\E[X_t]=0$.
    \end{assumpexo}
    
     Assumption \ref{ass:exo2} appears stronger than $X_t \ind U_{h,t+h}$, which is necessary for the result on linear LPs, but it might not be in practice. The nuisance $U_{h,t+h}$ will often either contain the state $S_{t-1}$ or lagged variables that are related to it. The extra content of the assumption only bites when $S_{t-1}$ carries shock-relevant information that is not in $U_{h,t+h}$. 
    
      \begin{remark}
      In typical time series models, lags or leads of $X_t$ or some functions of it will be in $U_{h,t+h}$. This implicitly requires that $X_t$ is independent over time. While shocks are usually assumed to be linearly unpredictable, independence is a substantive restriction. For example, it rules out conditional heteroskedasticity such as (G)ARCH errors, which are common in finance. Also, $X_t \ind S_{t-1}$ is not so innocent. It fails if the government spending shock $X_t$ is especially volatile or has fatter tails after a recession hits (think of stimulus packages and austerity).
      \end{remark}
        
    \begin{remark}
        The strong independence requirement, Assumption \ref{ass:exo2}, is necessary to allow for the structural function $\psi_h$ to be arbitrarily nonlinear. Otherwise, the equality 
        \begin{equation}\label{eq:condition_ind}
            \E_s[\psi_h(x,U_{h,t+h})] = \E[Y_{t+h} \mid X_t = x, S_{t-1}=s]
        \end{equation}
        generally does not hold since $U_{h,t+h}$ can depend on $X_t$ nonlinearly. With additional functional form assumptions, the independence requirement can be weakened. For example, with $U_{h,t+h} = \{S_{t-1},u_{h,t+h}\}$ and additive separability of the form 
        \begin{equation*}
            \psi_h(X_t,U_{h,t+h}) = \psi_{x,h}(X_t,S_{t-1}) + u_{h,t+h},
        \end{equation*}
        the assumption $\E[u_{h,t+h} \mid X_t,S_{t-1}] = 0$ is sufficient for \eqref{eq:condition_ind}. This highlights a tradeoff between assumptions about functional form and exogeneity that the researcher has to navigate.
    \end{remark}

    \begin{remark}
      If $X_t \ind (S_t, U_{h,t+h})$, all results of this paper apply to LPs where the shock is interacted with $S_t$ instead of $S_{t-1}$. Since the majority of the applied papers listed by \cite{Goncalves:24} interact with $S_{t-1}$ and the assumption $S_{t-1} \ind X_t$  is usually more plausible than $S_t \ind X_t$, this paper assumes that a lagged state is used. Conditioning on a state $S_t$ that is contemporaneously affected by the shock $X_t$ is left for future research. In this case, even defining the effect of $X_t$ on $Y_{t+h}$ conditional on $S_t$ is non-trivial since the shock moves the variable on which the effect is conditioned.\footnote{More generally, with an endogenous state $S_t$ that is affected by $X_t$ in the same period, not even unconditional effects might be defined without substantial restrictions on the data generating process. The reason is that equilibria might be non-existent or non-unique, so the outcome is not a well-defined function of the shock. For VARs with occasionally binding constraints see \cite{Mavroeidis:21} and \cite{Aruoba:22}; for rational expectations models with a zero lower bound see \cite{Ascari:22}. I thank an anonymous referee for alerting me to this point.}
    \end{remark}

      The following result shows that the state-dependent LP estimand is the projection coefficient of the average marginal effect $\theta_h(S_{t-1},\omega_X)$ on $f_{t-1}$.

      \begin{prop}\label{prop:1}
        Let Assumptions \ref{ass:exo2} and \ref{ass:sLP} (regularity conditions, in Appendix \ref{app:causal}) hold. Then the estimand $\beta^h$ of the state-dependent local projection \eqref{eq:state_dep_lp} has the following property:
        \begin{align}\label{eq:pop_regression_s}
        \beta^h & = \E \left[(f_{t-1}X_t) (f_{t-1}X_t)' \right]^{-1} \E\left[(f_{t-1}X_t) Y_{t+h} \right] \notag\\ 
        & = \E \left[f_{t-1} f_{t-1}'  \right]^{-1} \E \left[ f_{t-1} \theta_h(S_{t-1},\omega_X) \right],
        \end{align}
        where $\theta_h$ and $\omega_X$ are defined in \eqref{eq:av_marg_eff} and \eqref{eq:omega-x}.
        \end{prop}
        The proof is in Appendix~\ref{proof:lp}.

        Proposition~\ref{prop:1} generalizes the known identity \eqref{eq:lp_lin_effect} to cover virtually all state-dependent LP specifications used in practice. A formal result is needed because informal intuition is a poor guide to higher-order terms in a linear regression: \cite{White:80} showed that coefficients on polynomial terms generally do not correspond to Taylor coefficients of the structural function.\footnote{Similarly, LPs including nonlinear transformations of the shock $X_t$ are not straightforward to interpret causally; see Proposition 2 of \cite{Kolesar:24} for an example with the regressor $X_t^2$. \cite{Caravello:24} more generally show how nonlinear terms in $X_t$ can be used to disentangle sign and size nonlinearities of shock effects.}
          
        Supplemental Appendix S.3 illustrates and numerically verifies Proposition~\ref{prop:1} with a smooth transition VAR à la \cite{Auerbach:12} using an endogenous, continuous state variable.

        \paragraph*{Practical Implications} For the simple specification \eqref{eq:state_dep_simple} with a single state variable, the interaction term is just 
        \begin{equation*}
            \beta_1^h = \frac{\operatorname{Cov}[\theta_h(S_{t-1},\omega_X),S_{t-1}]}{\mathbb{V}[S_{t-1}]}.
        \end{equation*}
          It follows that state-dependent LPs can be interpreted like conventional linear regressions. Reading $\beta_1^h > 0$ as evidence for positive correlation of effect and state is valid regardless of the true data generating process. However, a reported impulse response estimate $\hat \beta_0^h + s \hat \beta_1^h$ can be biased for the effect $\theta_h(s,\omega_X)$ if the effect is nonlinear in the state.
            
          Commonly, two impulse responses $\hat \beta_0^h + s_L \hat \beta_1^h$ and $\hat \beta_0^h + s_H \hat \beta_1^h$ with $s_H \gg s_L$ are plotted \citep[e.g.][]{Auerbach:13}. Due to the possibility of extrapolation bias, such figures should be accompanied by robustness checks using different functional forms or a more flexible estimator of $\theta_h(s,\omega_X)$.\footnote{The MSE criterion that is minimized by linear regression prioritizes a good fit in regions of the state space $\mathcal{S}$ that occur frequently. This implies that effect estimands $\beta_0^h + s \beta_1^h$ for `extreme' $s$ might especially suffer from extrapolation bias.} Such a flexible estimator is proposed by \cite{David:26}, who fit the relationship between state and effect using splines. \cite{Auer:21} use polynomials. Supplemental Appendix S.2 discusses specific state-dependent LP specifications in more detail.
        
\subsection{Discussion: Is the LP Estimand Interesting?}

Proposition \ref{prop:1} shows that state-dependent local projections estimate the projection of a convex weighted average of marginal effects, even in the absence of parametric restrictions. This has two useful implications: First, this weighted average always lies in the convex hull of possible shock effects. This makes hypotheses such as `the fiscal multiplier is always larger in recessions than in expansions' testable. If the coefficient of the shock interacted with a recession indicator is significantly negative, the hypothesis is rejected. State-dependent local projections are frequently used for testing qualitative statements of this kind. Second, if the data generating process has no sign or size nonlinearities, i.e. $\Psi'_h(x,s)$ is constant in $x$, the LP estimand will be equal to \textit{any} weighted average of marginal effects. Both properties are valuable to have and make state-dependent LPs a natural `first column' \citep{Kolesar:25} specification if the true data generating process is uncertain.
   
Of course, not every causal estimand is automatically the object of interest. \cite{Goncalves:24} show that in a state-dependent VAR environment with endogenous state and Gaussian shocks, state-dependent LPs identify the effect of a marginal shock, which can be quite different from the potentially more interesting effect of larger shocks. If semi- or nonparametric estimators are precise enough, they are preferable to state-dependent LPs \citep{Goncalves:24b,Goncalves:26,Paranhos:25}. With the small effective sample sizes common in macroeconomics, however, this is not guaranteed.

To summarize, when the true model is uncertain, it is reassuring if an estimand lies within the range of possible effects and is exact when the shock enters linearly. The paper has so far shown that state-dependent LPs meet this standard while conventional state-dependent VAR estimands do not. The backshifted VAR developed in Section~\ref{sec:ssvar} matches the LP estimand and therefore shares its causal interpretation. Just as in the unconditional case, choosing between state-dependent LPs and the backshifted VAR is a finite-sample question.

\section{Simulation Study}\label{sec:sim}

The backshifted estimator only deserves a place in the applied toolkit if it improves on state-dependent LPs along any dimension. To assess this, I conduct a simple simulation study inspired by \cite{Li:24} but much smaller in scope.\footnote{To be precise, I simulate fewer DGPs than \cite{Li:24} and do not consider IV estimation, bias correction methods or estimation using integrated variables in levels.} The data generating process is assumed to be a state-dependent dynamic factor model with transition law 
\begin{equation*}
   \boldsymbol{f}_t = \sum_{\ell = 1}^4 \boldsymbol{\Gamma}_\ell(S_{t-\ell}) \boldsymbol{f}_{t-\ell} + \boldsymbol{H}(S_t) \boldsymbol{\varepsilon}_t.
\end{equation*}
The way $S_t$ enters allows for a similar kind of VAR bias as in the running example, but the data generating process also nests an ordinary dynamic factor model that is known to fit macroeconomic time series well. The DGP is estimated using the \cite{Stock:16} dataset with all variables transformed to stationarity. More details on the simulation study can be found in Supplemental Appendix S.4.
  
I consider two binary state variables indicating economic slack that are assumed to follow a Markov process: a high-unemployment state as in \cite{Ramey:18} and a recession state based on a moving average of GDP growth, as in \cite{Auerbach:12}. A cutoff is chosen such that the economy is in a slack state 30 percent of the time. This reflects the common characteristic that one state is rarer than the other. The high-unemployment state is quite persistent, while the low GDP growth state is faster mixing.\footnote{\cite{Ramey:18} note that the disagreement between state-dependent LPs and VARs might be higher for a transitory state.}
    
Similar to \cite{Li:24}, the LP and VAR estimators (one year of lags) are applied to six simulated series from the full macro panel.\footnote{All state-dependent methods include a state-specific constant, and the state-dependent LPs interact all regressors with the state, including the lags. Section~\ref{sec:empirics} uses the same estimators.} Since the six series follow a state-dependent VARMA process, the VAR methods are expected to carry truncation bias for longer horizons. Among the six variables, one is an observed government spending shock and one is the outcome of interest, a randomly chosen labor market variable.\footnote{Labor market variables tend to be quite persistent. Since I do not evaluate the response of the data in levels but only the response of a stationary transformation, focusing on labor market variables yields a compromise between fast- and slow-decaying responses.} For each state variable, I draw 100 such DGPs, for each of which I sample 1,000 time series of length $T=240$.

\begin{figure}[p]
  \centering
  \caption{Simulation Results: Bias, Standard Deviation and Coverage.}\label{fig:sim}
  \begin{subfigure}[t]{\textwidth}
    \centering
    \subcaption{Median absolute bias}\label{fig:sim_bias}
    \includegraphics[width=0.48\textwidth]{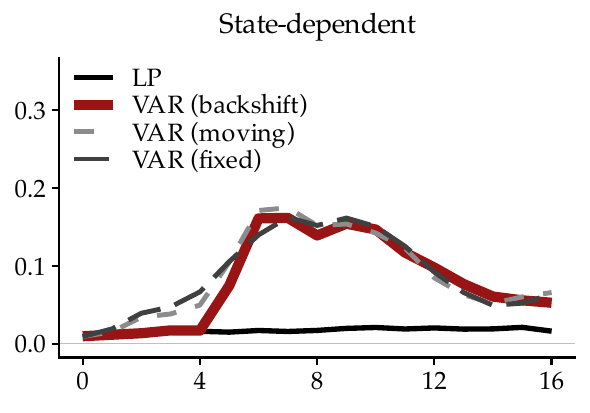}\hfill
    \includegraphics[width=0.48\textwidth]{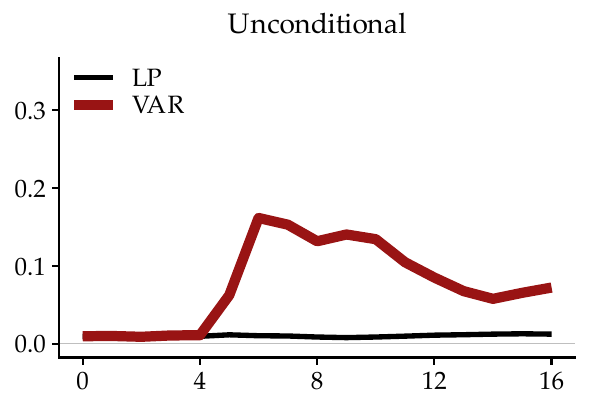}
  \end{subfigure}\\[1ex]
  \begin{subfigure}[t]{\textwidth}
    \centering
    \subcaption{Median standard deviation}\label{fig:sim_sd}
    \includegraphics[width=0.48\textwidth]{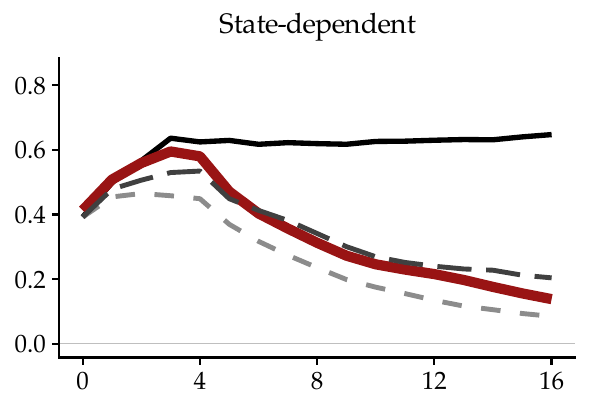}\hfill
    \includegraphics[width=0.48\textwidth]{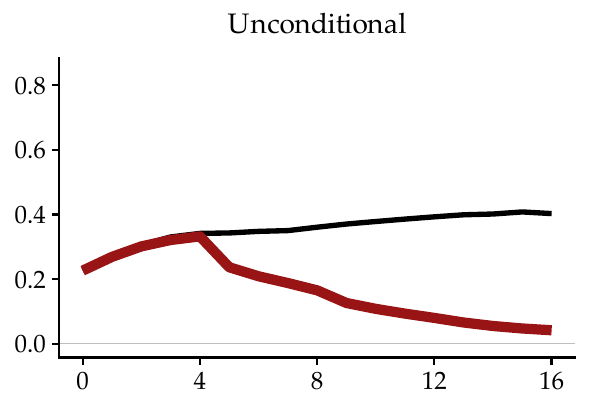}
  \end{subfigure}\\[1ex]
  \begin{subfigure}[t]{\textwidth}
    \centering
    \subcaption{Coverage of 90\% confidence intervals}\label{fig:sim_cov}
    \includegraphics[width=0.48\textwidth]{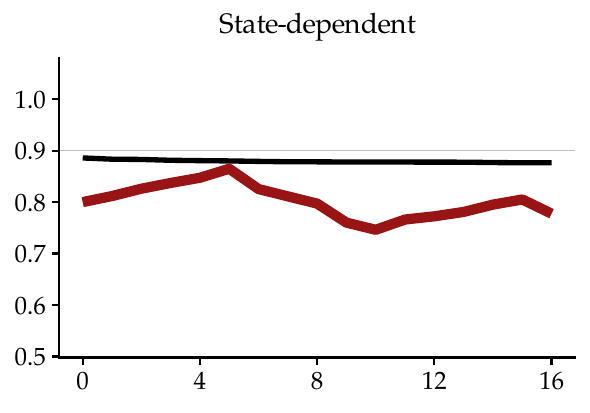}\hfill
    \includegraphics[width=0.48\textwidth]{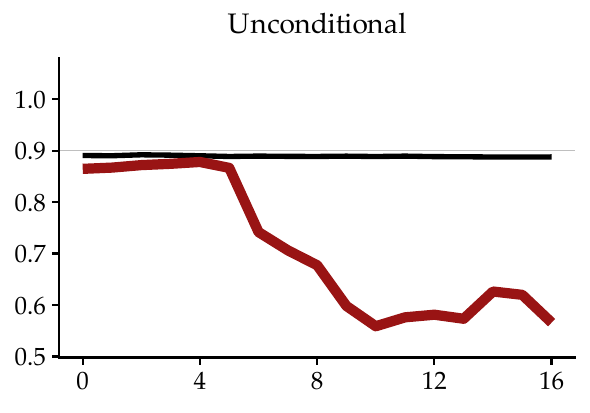}
  \end{subfigure}
  \notes{Left panels: state-dependent LP and VAR estimators, median (coverage: mean) over DGPs and both states. Right panels: linear LP and VAR, median (coverage: mean) over DGPs. (a) Absolute bias and (b) standard deviation of the estimates at each horizon, divided by the $L_2$ norm of the DGP's true unconditional impulse response over $h=0,\ldots,16$. (c) Share of samples in which the pointwise 90\% interval covers the true impulse response; heteroskedasticity-robust (HC1) standard errors for the LP and delta-method standard errors for the VAR, with the impact written as the coefficient on the shock in the regressions of the observed series on the shock and the lagged variables. Horizontal axis: horizon in quarters. 100 DGPs per state variable, 1,000 samples of length $T=240$ each, $p=4$ lags.}
\end{figure}

Figure \ref{fig:sim_bias} shows the median (across states and DGPs) absolute bias of the LP and VAR estimators. As perhaps expected, LPs outperform all the VAR estimators, especially at long horizons. Interestingly, as long as the horizon does not exceed the lag order, the state-dependent LP and backshifted VAR bias is almost exactly the same (left panel). This is the well-known pattern from the unconditional case (right panel). The fixed and moving state estimators, however, accumulate bias from early on. This already indicates that their relationship to LPs is fundamentally different from the unconditional case. For longer horizons, all VAR methods have similar bias.
    
Figure \ref{fig:sim_sd} shows the median standard deviation of all the methods. According to this metric, LPs are less precise than VAR methods, including the novel backshifted VAR. This might be counter-intuitive since the new method combines $h+1$ separate VARs, so in total a huge number of parameters has to be estimated. The solution to this puzzle is given by \cite{Ludwig:24}, who shows that LPs themselves are numerically equal to a sequence of iterated VAR projections, with each VAR increasing the lag order by one. The backshifted VAR iterates the projections \textit{without} increasing the lag order, leading to higher parsimony compared to the LP. This explains its lower variance.
   
There seems to be a similar finite-sample bias-variance tradeoff between state-dependent LPs and backshifted VARs as between unconditional LPs and VARs \citep{Li:24,Olea:24}. However, one important difference between the two settings should be highlighted: in this simulation study, the bias of state-dependent and unconditional estimators is of similar magnitude while state-dependent estimators generally have a higher variance. Trading off higher bias against lower variance is thus less costly in the state-dependent setting, both in terms of mean squared error and coverage. Figure \ref{fig:sim_cov} confirms this, plotting the average coverage of 90 percent confidence intervals. Due to their low bias, LPs achieve close to nominal coverage while at high horizons the unconditional VAR has a coverage of only around 60 percent. The backshifted VAR also has lower coverage than the LPs, but due to its higher variance it covers better than the unconditional VAR.\footnote{The backshifted VAR undercovers for the first 0--4 quarters since the delta-method standard errors underestimate the true variance. At longer horizons, the undercoverage of both VAR methods is due to bias of the point estimate; the standard error estimates work well for these horizons. See Supplemental Appendix S.4 for details.} State-dependent LPs remain the least biased estimator and deliver the most reliable coverage. But the backshifted VAR pays a smaller price for its precision than the unconditional VAR in terms of coverage.

\section{Empirical Application: Military Spending Shocks}\label{sec:empirics}

Lastly, I revisit the \cite{Ramey:18} study (from now on RZ) of state-dependent fiscal multipliers to compare state-dependent LPs with different VAR-based estimators. The shock $X_t$ is the military spending news shock of RZ in units of potential GDP. The outcome variables are quarterly government spending and GDP from 1890 to 2015, divided by potential GDP.\footnote{Potential GDP is taken as a sixth-order polynomial fitted to the GDP series. With this transformation, the cumulative multiplier is a ratio of cumulative impulse responses. With estimation in log-levels, the ratio would have to be adjusted by the government spending to GDP ratio, which varied heavily over the sample period.} The state variable $S_t$ indicates economic slack and is one if the unemployment rate is at or above 6.5 percent and zero otherwise.\footnote{This is the case in 36 percent of all quarters. The probability of staying in a no-slack state is 96 percent, for the slack state it is 93 percent.} All methods use one year of lags.
  
Figure \ref{fig:rz} plots the estimated impulse responses of government spending and GDP to a government spending news shock. The dashed lines are LP estimates with 68 percent confidence bands based on HC1 standard errors; the solid lines are VAR estimates with 68 percent delta-method confidence bands. 

\begin{figure}[p]
  \centering
  \caption{Responses of Government Spending and GDP to a Military News Shock.}\label{fig:rz}
  \begin{subfigure}[t]{\textwidth}
    \centering
    \subcaption{Slack}
    \includegraphics[width=0.32\textwidth]{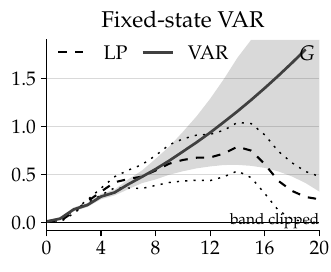}\hfill
    \includegraphics[width=0.32\textwidth]{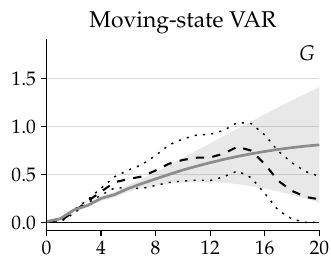}\hfill
    \includegraphics[width=0.32\textwidth]{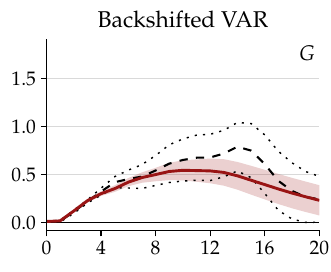}\\
    \includegraphics[width=0.32\textwidth]{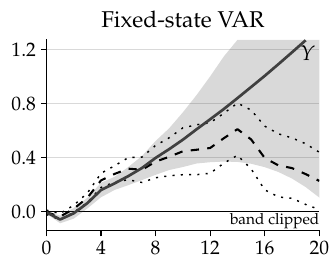}\hfill
    \includegraphics[width=0.32\textwidth]{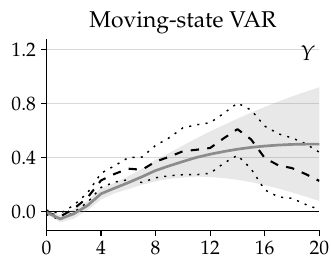}\hfill
    \includegraphics[width=0.32\textwidth]{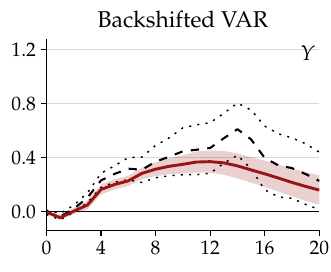}\\
  \end{subfigure}\\[1ex]
  \begin{subfigure}[t]{\textwidth}
    \centering
    \subcaption{No slack}
    \includegraphics[width=0.32\textwidth]{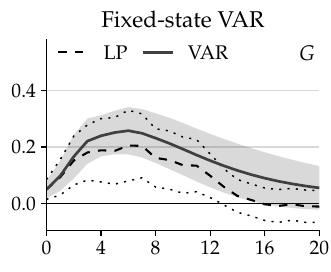}\hfill
    \includegraphics[width=0.32\textwidth]{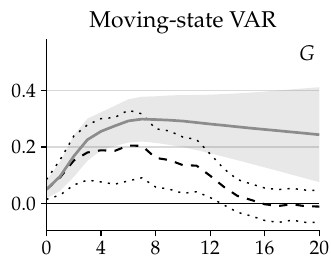}\hfill
    \includegraphics[width=0.32\textwidth]{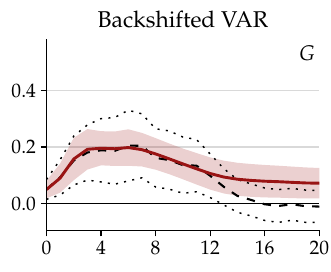}\\
    \includegraphics[width=0.32\textwidth]{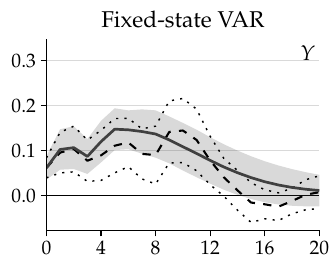}\hfill
    \includegraphics[width=0.32\textwidth]{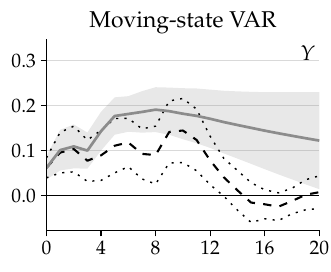}\hfill
    \includegraphics[width=0.32\textwidth]{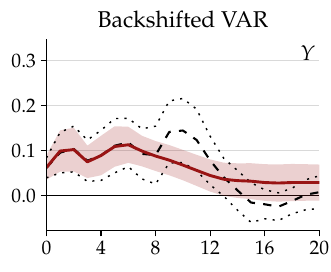}\\
  \end{subfigure}
  \notes{Responses of government spending ($G$) and GDP ($Y$) to a military news shock of one percent of potential GDP in the specification of \cite{Ramey:18}: 1890Q1--2015Q4, $p=4$ lags, slack state $S_{t-1} = \I[u_{t-1} \geq 6.5]$. Each panel shows the VAR estimator named in its title (solid line) with its pointwise 68 percent delta-method band (shaded), and the state-dependent LP (dashed) with its pointwise 68 percent band based on HC1 standard errors (dotted). Horizontal axis: quarters after the shock; vertical axis: percent of potential GDP. Within a row, panels share one scale; where the fixed-state VAR exceeds it, its band is clipped.}
\end{figure}

The left column shows the fixed-state VAR estimate. In the slack state, plotted in panel (a), it overshoots the LP confidence interval by multiple standard errors for longer horizons. Such a pattern was already noted by RZ, who argue that the fixed-state assumption is especially problematic in the slack state. Interestingly, both the output and government spending responses overshoot. This explains why the resulting VAR multiplier reported by RZ is close to the LP multiplier. Were it only the output response that overshoots, as RZ find to be the case in \cite{Auerbach:12}, the fixed-state VAR would yield a larger multiplier than LPs.
    
The second column shows the moving-state VAR, which accounts for state transitions using Markov transition probabilities that are estimated from $S_t$. It is closer to the LPs but also fails to match their decay after four years. Finally, the backshifted VAR matches the state-dependent LP estimates very closely for the first four quarters and afterwards captures their shape, but in a visually smoother way. This looks like a typical comparison of unconditional LPs and VARs.\footnote{There are cases where truncation bias can lead to large LP--VAR differences in the unconditional case. Truncation bias does not seem to be a big issue here.}
   
Panel (b) of Figure \ref{fig:rz} plots estimates for the no-slack state. The fixed-state VAR agrees well with the LPs, as is already found by RZ. However, the moving-state VAR, which accounts for the possibility of entering the slack state, is consistently higher than the LP estimate. Modeling the state transitions thus does not bring the VAR closer to the LP. This paper shows that this is not to be expected: since the VAR changes the projection operator at every prediction step, even correctly accounting for state transitions does not restore the LP estimand.\footnote{The model for the transition law of the state could be wrong too. However, modeling recessions as a Markov-switching process is not unheard of; see, for example, \cite{Hamilton:89}.} Again, the backshifted VAR looks like a smoothed version of the LP estimates.
  
In summary, in the setting of \cite{Ramey:18}, the backshifted VAR is the only VAR-based estimator that is consistently close to the state-dependent LP, while the other two VAR-based methods seem to have some pathologies. They also have substantially higher standard errors in the slack state, since they iterate on the explosive dynamics of the slack regime, which the backshifted VAR avoids.

\section{Conclusion}\label{sec:conclusion}

This paper showed in a nonparametric setting that state-dependent LPs have a causal estimand and conventional state-dependent VARs do not necessarily. The two methods are reconciled by means of a novel VAR-based estimator that has the same probability limit as the state-dependent LP estimator up to the lag order. 
  
According to a simulation study, the relationship between the backshifted VAR and state-dependent LPs is very similar to the unconditional LP--VAR relationship: choosing between the two methods amounts to a finite-sample bias-variance tradeoff. Since state-dependent estimates are generally noisier than unconditional estimates while having a similar bias, the backshifted VAR has better coverage than the unconditional VAR, even though state-dependent LPs remain clearly better in that regard. 

Several directions for future research remain open. First, it would be interesting to see if the qualitative relationship between state-dependent LPs and the backshifted VAR found in this paper continues to hold in a richer simulation study, considering a wider range of nonlinear data generating processes. The second point concerns inference. Even though for longer horizons the delta-method standard errors for the VAR were relatively accurate in the simulation study, this might not be the case for more persistent data. Bootstrap procedures might help improve finite-sample performance of the backshifted VAR. Third, researchers using nonlinear VARs have developed identification strategies specific to this setting \citep[e.g.][]{Mavroeidis:21}. Investigating whether some of these strategies carry over to state-dependent LPs would be insightful.

\clearpage
\appendix
\begin{appendix}
\section*{\LARGE Appendices}
\addcontentsline{toc}{section}{Appendices}

\section{Proof of Proposition \protect\ref{prop:var_lp}}\label{proof:recursive}

\begin{proof}
The proof shows by induction over $h$ that $\boldsymbol{Y}_{t+h}$ can be written in terms of $\boldsymbol{Y}_{t-1}$ and the orthogonalized projection errors $\boldsymbol{E}_{t+l}^{l,\perp}$, $l = 0, \ldots, h$, of the projections in \eqref{eq:vars_decreasing}, with coefficients that depend on the state only through $S_{t-1}$:
\begin{align*}
    \boldsymbol{Y}_{t+h} & = \tilde{\boldsymbol{\Pi}}^h(S_{t-1}) \boldsymbol{Y}_{t-1} + \sum_{l=0}^{h} \tilde{\boldsymbol{A}}_l^h(S_{t-1}) \boldsymbol{E}_{t+l}^{l,\perp}.
\end{align*}
This is possible because the projection with index $h$ in \eqref{eq:vars_decreasing}, applied to $\boldsymbol{Y}_{t+h}$, conditions on $S_{(t+h)-1-h} = S_{t-1}$.

For $h=0$, the first projection in \eqref{eq:vars_decreasing} is of this form with $\tilde{\boldsymbol{\Pi}}^0 \equiv \boldsymbol{\Pi}^0$ and $\tilde{\boldsymbol{A}}_0^0 \equiv \boldsymbol{A}^0$. Now suppose the desired representation exists for $\boldsymbol{Y}_{t+h-1}$. Then use the projection with index $h$ in \eqref{eq:vars_decreasing}:
\begin{align*}
    \boldsymbol{Y}_{t+h} & = \boldsymbol{\Pi}^h(S_{t-1}) \boldsymbol{Y}_{t+h-1} + \boldsymbol{A}^h(S_{t-1}) \boldsymbol{E}_{t+h}^{h,\perp}.
\end{align*}
Use the induction hypothesis, dropping the dependence of the parameters on $S_{t-1}$:
\begin{align*}
    \boldsymbol{Y}_{t+h} & = \underbrace{\boldsymbol{\Pi}^h \tilde{\boldsymbol{\Pi}}^{h-1}}_{\tilde{\boldsymbol{\Pi}}^h \equiv} \boldsymbol{Y}_{t-1} + \sum_{l=0}^{h-1} \underbrace{\boldsymbol{\Pi}^h \tilde{\boldsymbol{A}}_l^{h-1}}_{\tilde{\boldsymbol{A}}_l^h \equiv} \boldsymbol{E}_{t+l}^{l,\perp} + \underbrace{\boldsymbol{A}^h}_{\tilde{\boldsymbol{A}}_h^h \equiv} \boldsymbol{E}_{t+h}^{h,\perp}.
\end{align*}
This is of the desired form. In particular, $\tilde{\boldsymbol{A}}_0^h = \boldsymbol{\Pi}^h \cdot \ldots \cdot \boldsymbol{\Pi}^1 \boldsymbol{A}^0$, so by \eqref{eq:theta_b}
\begin{equation*}
    \theta^b_{\mathit{VAR,h}}(s) = (\tilde{\boldsymbol{A}}_0^h(s))_{2,1}.
\end{equation*}

For $1 \leq k \leq \min(h,p)$, the stacked vector $\boldsymbol{Y}_{t+k-1}$ contains $X_t$, so by construction of \eqref{eq:vars_decreasing}
\begin{equation*}
    \E[\boldsymbol{E}_{t+k}^{k,\perp} X_t \mid S_{t-1}] = \mathbf{0}.
\end{equation*}
For $p < k \leq h$, this holds at $S_{t-1} = s$ by the condition of the proposition, evaluated at date $t+k$. Moreover, by construction,
\begin{equation*}
    \E[e_t^{0,\perp} X_t \mid S_{t-1}] = 0.
\end{equation*}
Since by assumption
\begin{equation*}
    \E[\boldsymbol{Y}_{t-1} X_t \mid S_{t-1}] = \mathbf{0},
\end{equation*}
it follows from the representation derived above that
\begin{equation*}
    (\tilde{\boldsymbol{A}}_{0}^h(s))_{2,1} X_t = P_s[Y_{t+h} \mid X_t],
\end{equation*}
which finishes the proof.
\end{proof}

\section{Causal Interpretation of Local Projections: Details}\label{app:causal}

This appendix collects additional formal assumptions and results underlying the causal interpretation of LPs discussed in Section \ref{sec:main}. I first present Proposition 1 of \cite{Kolesar:24}, which is used in the proof of Proposition \ref{prop:1} of this paper, and then state the regularity conditions needed for Proposition \ref{prop:1}.

\paragraph*{Linear Local Projections} The causal interpretation of a linear local projection rests on the following regularity and exogeneity conditions.

\begin{assumpLP}\label{ass:lp_reg}
(i) The data $(Y_{t+h},X_t)$ are covariance stationary and $X_t$ is continuously distributed on an interval $I \subseteq \mathbb{R}$ with strictly positive variance. (ii) The conditional mean $g_h(x) \equiv \E[Y_{t+h} \mid X_t = x]$ is locally absolutely continuous on $I$. (iii) Finally, $\int_I \omega_X(x) |g_h'(x)| dx < \infty$, where the weights $\omega_X$ are defined in \eqref{eq:omega-x}.
\end{assumpLP}

\begin{assumpexoLP}\label{ass:lp_exo}
For all $h \geq 0$, $t \in \mathbb{Z}$, $X_t \ind U_{h,t+h}$ and $\E[X_t] = 0$.
\end{assumpexoLP}

Assumption \ref{ass:lp_reg} is a collection of regularity conditions ensuring that the LP estimand is well-defined, the conditional mean function $g_h$ has a derivative almost everywhere and a specific weighted average of the derivative is finite. Assumption \ref{ass:lp_exo} requires the shock $X_t$ and the other variables entering $Y_{t+h}$ to be independent. This ensures that the conditional mean function $g_h$ nonparametrically identifies the average structural function $\Psi_h$, so that the derivative of $g_h$ has a causal interpretation.

The following result is part of Proposition 1 of \cite{Kolesar:24} and is a formal statement of identity \eqref{eq:lp_lin_effect} in the main text:

\begin{lem}[\citeauthor{Kolesar:24},~\citeyear{Kolesar:24}]\label{lem:lp_lin}
Suppose Assumptions \ref{ass:lp_reg} and \ref{ass:lp_exo} hold. Then the local projection estimand satisfies
\begin{equation*}
\frac{\E[Y_{t+h}X_t]}{\E[X_t^2]} = \int \omega_X(x)\Psi_h'(x,\mathcal{S}) dx = \theta_h(\mathcal{S},\omega_X).
\end{equation*}
\end{lem}

The proof of Lemma \ref{lem:lp_lin} amounts to an application of the fundamental theorem of calculus and Fubini's theorem. Its key implication is that, even when the data generating process is far from linear, a linear LP still estimates a proper weighted average of the marginal causal effects $\Psi_h'(\cdot)$.

\paragraph*{State-Dependent Local Projections} For the state-dependent results in the main text, Assumption \ref{ass:lp_reg} is replaced by the following analogue, which guarantees that Lemma \ref{lem:lp_lin} applies to the conditional measure given $S_{t-1}$.

\begin{assumpsLP}\label{ass:sLP}
(i) The data $(Y_{t+h}, f_{t-1}',X_t)$ are covariance stationary, the matrix $\E[f_{t-1} f_{t-1}']$ is positive definite and $X_t$ is continuously distributed on an interval $I \subseteq \mathbb{R}$ with strictly positive variance. Furthermore, for almost every $s \in \mathcal{S}$, (ii) the function $g_h(x,s) \equiv \E[Y_{t+h} \mid X_t = x, S_{t-1} = s]$ is locally absolutely continuous on $I$ in $x$, and (iii) $\int_I \omega_X(x) | \partial_x g_h(x,s)| dx < \infty$, where the weights $\omega_X$ are defined in \eqref{eq:omega-x}.
\end{assumpsLP}

The exogeneity condition accompanying Assumption \ref{ass:sLP} is Assumption \ref{ass:exo2}, which strengthens Assumption \ref{ass:lp_exo} by additionally requiring the lagged state $S_{t-1}$ to be independent of the shock $X_t$. A detailed discussion is provided in Section \ref{sec:main}.

\section{Proof of Proposition \protect\ref{prop:1}}\label{proof:lp}

\begin{proof}
        Consider the population normal equation of the regression \eqref{eq:state_dep_lp} and apply the law of iterated expectations:
        \begin{align*}
        0 & = \E[X_t f_{t-1}(Y_{t+h} - X_t f_{t-1}'\beta^h)] \\
        & = \E[f_{t-1} \E[Y_{t+h} X_t - X_t^2 f_{t-1}'\beta^h \mid S_{t-1}]] \\
        & = \E \left[ \E[X_t^2 \mid S_{t-1}] f_{t-1} \left( \frac{\E[Y_{t+h} X_t \mid S_{t-1}]}{\E[X_t^2 \mid S_{t-1}]} - f_{t-1}'\beta^h \right) \right].
        \end{align*}
        Due to independence, $\E[X_t^2 \mid S_{t-1}] = \E[X_t^2]$. This term can be pulled out of the expectation and dropped. Now Proposition 1 of \cite{Kolesar:24} (Lemma~\ref{lem:lp_lin}) can be applied to the conditional measure:
        \begin{align*}
        \frac{\E[Y_{t+h} X_t \mid S_{t-1}]}{\E[X_t^2 \mid S_{t-1}]} & = \int \frac{\E[\I[X_t \geq x]X_t \mid S_{t-1}]}{\E[X_t^2 \mid S_{t-1}]} \Psi'_h(x, S_{t-1}) dx \\
        & = \theta_h(S_{t-1},\omega_X),
        \end{align*}
        since due to independence of $X_t$ and $S_{t-1}$ the weights on $\Psi_h'(x,S_{t-1})$ do not depend on the state. Therefore, the above normal equation yields
        \begin{equation*}
        0 = \E[f_{t-1} (\theta_h(S_{t-1},\omega_X) - f_{t-1}'\beta^h)],
        \end{equation*}
        which is the normal equation of the population regression \eqref{eq:pop_regression_s}.
        \end{proof}

\end{appendix}

\clearpage
\hypersetup{pageanchor=false}   %
\setcounter{page}{1}
\setcounter{section}{0}
\setcounter{figure}{0}
\setcounter{table}{0}
\setcounter{equation}{0}
\setcounter{prop}{0}
\renewcommand{\thesection}{S.\arabic{section}}
\renewcommand{\thefigure}{S.\arabic{figure}}
\renewcommand{\thetable}{S.\arabic{table}}
\renewcommand{\theequation}{S.\arabic{equation}}
\renewcommand{\theprop}{S.\arabic{prop}}
\renewcommand{\theHsection}{S.\arabic{section}}
\renewcommand{\theHfigure}{S.\arabic{figure}}
\renewcommand{\theHtable}{S.\arabic{table}}
\renewcommand{\theHequation}{S.\arabic{equation}}
\renewcommand{\theHprop}{S.\arabic{prop}}

\begin{center}
  {\LARGE Supplemental Appendix\par}
  \vspace{0.6cm}
  {\large State-Dependent Local Projections and Vector Autoregressions\par}
  \vspace{0.3cm}
  {\normalsize Valentin Winkler\par}
\end{center}
\vspace{1cm}

\etocdepthtag.toc{online}
\etocsettagdepth{main}{none}
\etocsettagdepth{online}{subsection}
\etocsettocstyle{\section*{Contents}}{}
\tableofcontents
\clearpage

\section{Applied Literature Using State-Dependent LPs}\label{app:applied}
Table \ref{tab:lit} summarizes the 44 applied studies using state-dependent LPs that are surveyed by \citesupp{Goncalves:24}. This shows that the majority of the studies (29/44) use a lagged state variable, which is the specification considered in the main paper. Also, only 15/44 studies restrict themselves to one binary state variable (some studies interact multiple binary state variables), while 19/44 use a continuous state variable. The results of the main paper are broad enough to cover virtually all estimating equations and state variables used in the applied studies. Table \ref{tab:lit_full} lists all papers individually.

\begin{table}[ht!]
  \centering
  \caption{Some Properties of the 44 Studies Listed by \protect\citesupp{Goncalves:24}.}\label{tab:lit}
  \begin{tabular}{lccc}\toprule
      Category & State Lagged & State Binary & State Continuous \\ \midrule
      Fiscal Policy &  25/27 & 13/27 & 8/27 \\
      Monetary Policy & 4/11 & 2/11 & 7/11 \\
      Market Reforms & 0/3 & 0/3 & 2/3 \\
      Other & 0/3 & 0/3 & 2/3 \\ \midrule 
      Total & 29/44 & 15/44 & 19/44 \\ \bottomrule
  \end{tabular}
     
    \notes{The 44 papers listed here are from \citesupp{Goncalves:24}. The properties are from my own reading of the papers. The column `State Lagged' indicates whether the shock is interacted with (a function of) $S_{t-1}$. Otherwise, the contemporaneous state $S_t$ is used. The column `State Binary' indicates whether the paper only considers a specification with $S_t$ being a binary dummy variable, while `State Continuous' marks whether a continuous state variable is used. For this summary, only the main text of the papers and not their appendices is considered.}
\end{table}

\begingroup
\setlength{\tabcolsep}{3pt}
\begin{center}
\begin{longtable}{@{} p{0.44\linewidth} *{2}{>{\centering\arraybackslash}p{0.12\linewidth}} >{\centering\arraybackslash}p{0.16\linewidth} >{\centering\arraybackslash}p{0.09\linewidth} @{}}
\caption{Some Properties of the Studies Listed by \protect\citesupp{Goncalves:24}.}\label{tab:lit_full} \\
\toprule
{Paper} &
{\shortstack{$S_t$\\Lagged}} &
{\shortstack{$S_t$\\Binary}} &
{\shortstack{$S_t$\\Continuous}} &
{2SLS} \\
\midrule
\endfirsthead

\multicolumn{5}{c}%
{\tablename\ \thetable\ -- \textit{Continued}} \\
\toprule
{Paper} &
{\shortstack{$S_t$\\Lagged}} &
{\shortstack{$S_t$\\Binary}} &
{\shortstack{$S_t$\\Continuous}} &
{2SLS} \\
\midrule
\endhead

\bottomrule
\endfoot

\bottomrule
\endlastfoot

\multicolumn{5}{l}{\textbf{Fiscal Policy}} \\
\citesupp{Alloza:22} & \ding{51} & \ding{51} & \ding{55} & \ding{55} \\
\citesupp{Auerbach:13b} & \ding{51} & \ding{55} & \ding{51} & \ding{55} \\
\citesupp{Auerbach:16} & \ding{51} & \ding{55} & \ding{55} & \ding{55} \\
\citesupp{Ben-Zeev:23} & \ding{55} & \ding{51} & \ding{55} & \ding{51} \\
\citesupp{Berge:21} & \ding{51} & \ding{55} & \ding{55} & \ding{51} \\
\citesupp{Bernardini:18} & \ding{51} & \ding{55} & \ding{55} & \ding{51} \\
\citesupp{Bernardini:20} & \ding{51} & \ding{55} & \ding{51} & \ding{51} \\
\citesupp{Biolsi:17} & \ding{51} & \ding{51} & \ding{55} & \ding{55} \\
\citesupp{Boehm:20} & \ding{51} & \ding{51} & \ding{55} & \ding{55} \\
\citesupp{Born:20} & \ding{51} & \ding{55} & \ding{51} & \ding{55} \\
\citesupp{Choi:22} & \ding{51} & \ding{51} & \ding{55} & \ding{55} \\
\citesupp{Demirel:21} & \ding{51} & \ding{55} & \ding{51} & \ding{51} \\
\citesupp{El-Shagi:21} & \ding{51} & \ding{55} & \ding{51} & \ding{55} \\
\citesupp{Eminidou:23} & \ding{51} & \ding{55} & \ding{51} & \ding{51} \\
\citesupp{Ghassibe:22} & \ding{51} & \ding{55} & \ding{55} & \ding{51} \\
\citesupp{Jorda:16} & \ding{51} & \ding{55} & \ding{55} & \ding{51} \\
\citesupp{Klein:17} & \ding{51} & \ding{55} & \ding{55} & \ding{55} \\
\citesupp{Klein:21} & \ding{51} & \ding{51} & \ding{55} & \ding{51} \\
\citesupp{Klein:22} & \ding{51} & \ding{51} & \ding{55} & \ding{51} \\
\citesupp{Leduc:12} & \ding{51} & \ding{55} & \ding{51} & \ding{55} \\
\citesupp{Liu:22} & \ding{51} & \ding{51} & \ding{55} & \ding{51} \\
\citesupp{Liu:23} & \ding{51} & \ding{51} & \ding{55} & \ding{51} \\
\citesupp{Miyamoto:18} & \ding{51} & \ding{51} & \ding{55} & \ding{51} \\
\citesupp{Owyang:13} & \ding{51} & \ding{51} & \ding{55} & \ding{55} \\
\citesupp{Ramey:18} & \ding{51} & \ding{51} & \ding{55} & \ding{51} \\
\citesupp{Riera:15} & \ding{51} & \ding{55} & \ding{51} & \ding{55} \\
\citesupp{Sheremirov:22} & \ding{55} & \ding{51} & \ding{55} & \ding{51} \\

\\[0.2ex]
\multicolumn{5}{l}{\textbf{Monetary Policy}} \\
\citesupp{Albrizio:20} & \ding{55} & \ding{55} & \ding{51} & \ding{51} \\
\citesupp{Albuquerque:19} & \ding{51} & \ding{55} & \ding{55} & \ding{55} \\
\citesupp{Alpanda:21} & \ding{51} & \ding{55} & \ding{55} & \ding{55} \\
\citesupp{Auer:21} & \ding{51} & \ding{55} & \ding{51} & \ding{55} \\
\citesupp{El-Herradi:21} & \ding{55} & \ding{51} & \ding{55} & \ding{51} \\
\citesupp{Falck:21} & \ding{51} & \ding{55} & \ding{51} & \ding{55} \\
\citesupp{Furceri:18} & \ding{55} & \ding{55} & \ding{51} & \ding{55} \\
\citesupp{Jorda:20} & \ding{55}  & \ding{51} & \ding{55} & \ding{51} \\
\citesupp{Santoro:14} & \ding{55} & \ding{55} & \ding{51} & \ding{51} \\
\citesupp{Tenreyro:16} & \ding{55} & \ding{55} & \ding{51} & \ding{55} \\
\citesupp{Tillmann:20} & \ding{55} & \ding{55} & \ding{51} & \ding{55} \\

\\[0.2ex]
\multicolumn{5}{l}{\textbf{Market Reforms}} \\
\citesupp{Haan:22} & \ding{55} & \ding{55} & \ding{55} & \ding{55} \\
\citesupp{Duval:18} & \ding{55} & \ding{55} & \ding{51} & \ding{55} \\
\citesupp{Lastauskas:20} & \ding{55} & \ding{55} & \ding{51} & \ding{55} \\

\\[0.2ex]
\multicolumn{5}{l}{\textbf{Other}} \\
\citesupp{Cacciatore:21} & \ding{55} & \ding{55} & \ding{51} & \ding{55} \\
\citesupp{Loipersberger:22} & \ding{55} & \ding{55} & \ding{55} & \ding{55} \\
\citesupp{Sheng:21} & \ding{55} & \ding{55} & \ding{51} & \ding{55} \\

\multicolumn{5}{l}{} \\
Total \ding{51} & 29 & 15 & 19 & 19 \\ \midrule

\multicolumn{5}{@{}p{\linewidth}@{}}{%
\legend{The 44 papers listed here are from \citesupp{Goncalves:24}. The column `$S_t$ Lagged' indicates whether the shock is interacted with (a function of) $S_{t-1}$. Otherwise, the contemporaneous state $S_t$ is used. This definition is silent about whether the state is forward-looking or exogenous/endogenous. The column `$S_t$ Binary' is checked if the paper only considers a specification with $S_t$ being a binary dummy variable, while `$S_t$ Continuous' is checked if a continuous state variable is used. Lastly, `2SLS' is checked if some kind of two-stage least squares estimator is computed in the paper. For this summary, only the main text of the papers and not their appendices is considered.}%
}\\

\end{longtable}
\end{center}
\endgroup

\section{Common State-Dependent LP Specifications}\label{app:specific}
      
Proposition 2 of the main paper shows that running a state-dependent local projection of the form (12) yields the same estimand as regressing the unobserved average conditional effect $\theta_h(S_{t-1},\omega_X)$ on the interaction term $f(S_{t-1})$. I use this insight to derive the causal estimand of common state-dependent LP specifications.

\begin{specification}[Binary State]
  Let $S_{t-1} \in \{0,1\}$ and consider the regression
  \begin{equation*}
  Y_{t+h} =  X_t \beta^h_0 +  S_{t-1} X_t \beta_1^h + \textnormal{error}_{h,t+h}.
  \end{equation*}
  It follows from Proposition 2 that the estimands satisfy
  \begin{equation*}
  \beta_0^h = \theta_h(0,\omega_X),~~~ \beta_1^h = \theta_h(1,\omega_X) - \theta_h(0,\omega_X).
  \end{equation*}
  If $\beta_1^h \neq 0$, the effect of $X_t$ on $Y_{t+h}$ is commonly interpreted as depending on the state $S_{t-1}$. This is justified since the interaction captures the difference between average conditional effects formed with the same weights. If the effect of $X_t$ is larger in state 1 than in state 0 across all baseline shock levels $x$, the non-negativity of $\omega_X$ ensures $\beta_1^h > 0$. Conversely, $\beta_1^h \neq 0$ implies state-dependence at some level $x$.
\end{specification}

For continuous states, Section 4.2 of the main paper shows that a linear interaction recovers a scaled covariance of state and effect, and recommends checking the implied impulse responses against more flexible specifications of the relationship between state and effect. The following two specifications are such flexible alternatives.

\begin{specification}[Series Expansion]
  \citesupp{Auer:21} address nonlinearities in the relationship between the state and the conditional effect by interacting $X_t$ with a polynomial basis in the state, i.e.
\begin{equation*}
Y_{t+h} = \sum_{j=0}^{J} S_{t-1}^{j} X_t \beta_{j}^h + \text{error}_{h,t+h},
\end{equation*}
with some degree $J \geq 1$. Proposition 2 shows that the estimand satisfies
\begin{equation*}
\beta^h = (\beta_{0}^h,\ldots,\beta_{J}^h)' = \arg \min_{b \in \mathbb{R}^{J+1}} \E \left[ \theta_h(S_{t-1},\omega_X) - \sum_{j=0}^{J} S_{t-1}^j b_j \right]^2.
\end{equation*}

Standard concerns about polynomial series regression apply \citepsupp[e.g.][Chapter 20]{Hansen:22}. Higher-order polynomials can behave erratically, especially at extreme values of $s$, which are often of interest in macroeconomics.
\end{specification}

\begin{specification}[Kernel Weighting]
   The natural cousin of series regression is kernel regression. Kernel weights are often used informally when a continuous state is discretized, which amounts to running an LP on a sample that is re-weighted with a rectangular kernel.\footnote{Some example studies and state variables are \citesupp{Jorda:16} (cyclical component of GDP), \citesupp{Klein:17} (private debt), \citesupp{Ramey:18} (unemployment) and \citesupp{Ghassibe:22} (unemployment $\times$ inflation).}

   Proposition 2 can be used to characterize the causal estimand of such weighted LPs more generally. Define $w(s) \equiv K(a^{-1} \lVert s-s^* \rVert)$, where $K$ is some kernel function and $a$ is a bandwidth. The OLS regression
   \begin{equation*}
      \sqrt{w(S_{t-1})} Y_{t+h} = \sqrt{w(S_{t-1})} X_t \beta^h(s^*) + \text{error}_{h,t+h}
   \end{equation*}
   has the estimand
   \begin{equation*}
      \beta^h(s^*) = \E \left[ \frac{w(S_{t-1})}{\E[w(S_{t-1})]} \theta_h(S_{t-1},\omega_X) \right],
   \end{equation*}
   which is the Nadaraya--Watson kernel regression of the conditional average effect $\theta_h(S_{t-1},\omega_X)$ on the state using the weight $w$. 

   This relationship can be used to improve on the practice of discretizing continuous states. If the average value of $\theta_h(S_{t-1},\omega_X)$ on the discretized bins is the ultimate object of interest, a rectangular kernel should be used. However, if the state is only discretized to avoid making functional form assumptions, non-rectangular kernels such as the Gaussian or triangular are generally a better choice, for example, to report impulse responses for a high and a low value of the state.\footnote{This distinction is related to \citeauthorsupp{Cattaneo:24}'s (\citeyearsupp{Cattaneo:24}) work on binscatter plots. Interpreting discretization as a kernel method also implies that a locally linear version of the kernel regression might be preferable to a local constant regression when approximating a function value $\theta_h(s^*, \omega_X)$.}
     
  To estimate kernel-weighted LPs, researchers can rely on the vast literature on varying coefficient models in statistics \citepsupp[see][for a survey]{Fan:08}. The main paper provides an identification result that shows how to interpret the resulting estimates without relying on a semiparametric model.
\end{specification}

\section{Illustrative Example: Smooth Transition VAR}\label{app:stvar}
Here I study state-dependent LPs when the data generating process is a smooth transition VAR à la \citesupp{Auerbach:12}, who used this model to study state-dependent government spending multipliers. The setup allows for an analytical computation of marginal effects while featuring a rich heterogeneity of effect sizes and a continuous state variable. This allows me to demonstrate the full scope of Proposition 2 of the main paper in a simulation study. 

\paragraph*{Data Generating Process} I follow \citesupp{Auerbach:12} as closely as possible. The number of endogenous variables and shocks is set to $n = 3$. The vector of endogenous variables observed at quarterly frequency is $\boldsymbol{Y}_t = (G_t, T_t, Y_t)'$, where $G_t$ is government purchases, $T_t$ is taxes and $Y_t$---the variable of interest---is real GDP. The shock of interest is a government spending shock $X_t$, which is the first element of the three-dimensional shock vector $\boldsymbol{\varepsilon}_t$. The state $S_{t}$ is defined as a demeaned and standardized average of GDP growth between $t-6$ and $t$ and serves as an indicator of past economic activity.\footnote{In this, I depart from \citesupp{Auerbach:12}, who use a centered moving average of GDP growth. Since such a state would not satisfy $X_t \ind S_{t-1}$ (government spending shocks affect current and future GDP growth), Proposition 2 could not be applied. Because of this modification, my later estimation results are qualitatively very different from \citeauthorsupp{Auerbach:12}'s (\citeyearsupp{Auerbach:12}). Sensitivity to the averaging window is a known feature of this model \citepsupp[see][]{Alloza:22} but should not concern us for the sake of this exercise.} The structural model is then given by
\begin{subequations}
\begin{align}
\boldsymbol{Y}_t & = \boldsymbol{c}(S_{t-1}) + \sum_{k=1}^p \boldsymbol{\Pi}_k(S_{t-1}) \boldsymbol{Y}_{t-k} + \boldsymbol{A}(S_{t-1}) \boldsymbol{\varepsilon}_t \label{eq:stvar_first} \\
\boldsymbol{\varepsilon}_t & \sim N(\mathbf{0}, \boldsymbol{I}_3)  \\
\boldsymbol{c}(S_{t-1}) & = (1-F(S_{t-1})) \boldsymbol{c}_{E} + F(S_{t-1}) \boldsymbol{c}_{R} \\
\boldsymbol{\Pi}_k(S_{t-1}) & = (1-F(S_{t-1})) \boldsymbol{\Pi}_{k,E} + F(S_{t-1}) \boldsymbol{\Pi}_{k,R} \\
\boldsymbol{A}(S_{t-1}) & = \operatorname{chol}((1-F(S_{t-1})) \boldsymbol{\Omega}_E + F(S_{t-1}) \boldsymbol{\Omega}_R) \label{eq:chol_A} \\
F(S_{t-1}) & = (1 + \exp(\gamma S_{t-1}))^{-1}, ~~~ \gamma > 0. \label{eq:stvar_last}
\end{align}
\end{subequations}

The intercepts and slope coefficients are convex combinations of $\boldsymbol{c}_{E}$ and $\boldsymbol{c}_{R}$ and of $\boldsymbol{\Pi}_{k,E}$ and $\boldsymbol{\Pi}_{k,R}$, respectively. If $S_{t-1}$ is low---the economy has been in a recession---$F(S_{t-1})$ is close to 1 and the parameters are shifted towards $\boldsymbol{c}_{R}$ and $\boldsymbol{\Pi}_{k,R}$; in the opposite case they are shifted towards $\boldsymbol{c}_{E}$ and $\boldsymbol{\Pi}_{k,E}$. The specification \eqref{eq:chol_A} allows a researcher to identify the structural shocks recursively from the variance--covariance matrix of the reduced-form errors.
  
\paragraph*{Estimation} I estimate the model \eqref{eq:stvar_first}--\eqref{eq:stvar_last} using the R package of \citesupp{sstvars} with $\gamma = 1.5$ and $p=3$ lags. These values and the data are from \citesupp{Auerbach:12}.\footnote{Their replication package can be found online: \url{https://www.openicpsr.org/openicpsr/project/114783/version/V1/view}.} While they estimate their system in log-levels, I log-difference the data, which is necessary to get a stationary distribution of causal effects.
  
\begin{figure}[t]
\caption{Distribution of Dynamic Shock Effects and State-Dependent LP Estimates.} \label{fig:effs_LP}
\justify
\begin{overpic}[scale=0.7, unit=1mm]{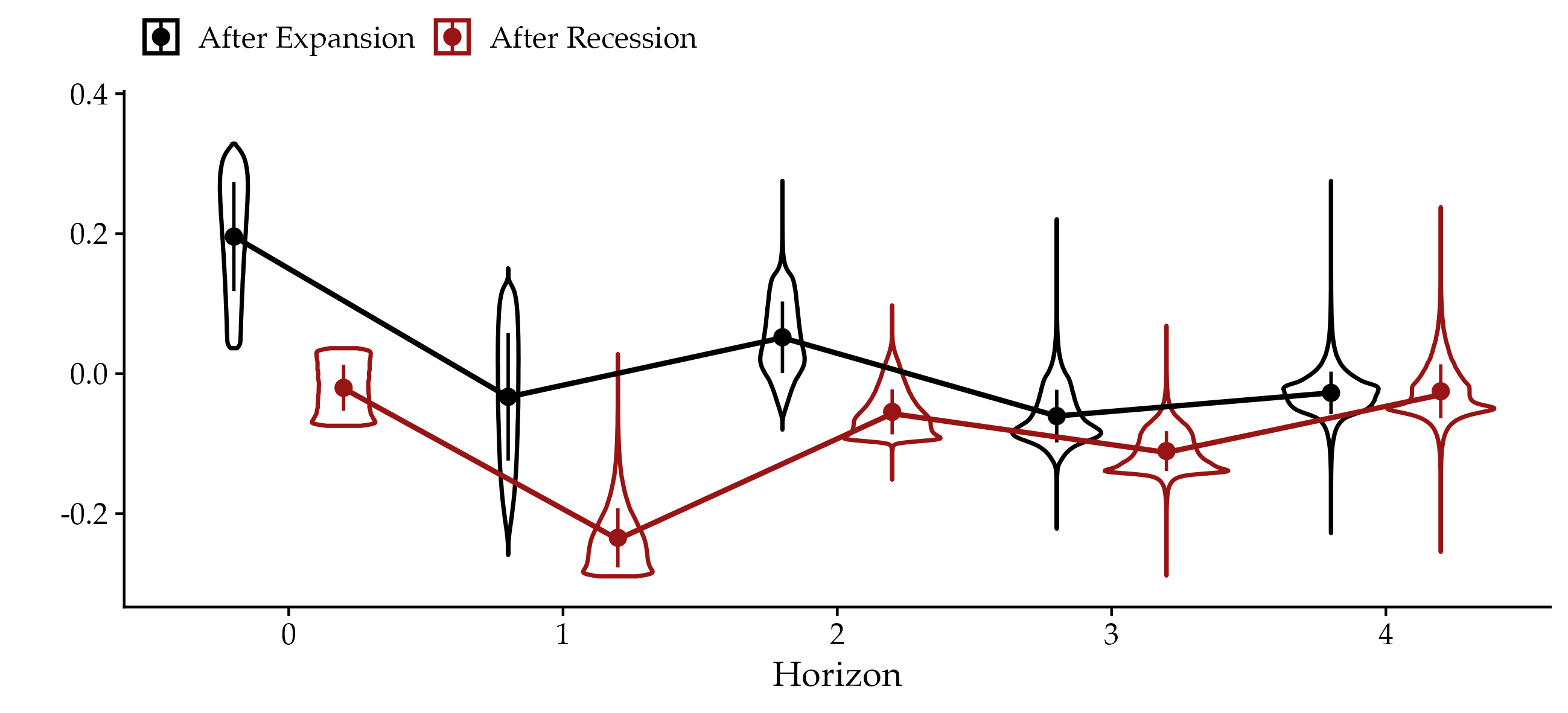}
\put(11,30.3){\line(1,0) {3.2}}
\put(20.5,35){\line(0.1,-1) {1.15}}
\put(17,36.5){$\psi_0'(X_t)$}
\put(15.6,35.5){\line(1,0.8) {1.3}}
\put(8,29.3){$\beta_0^0$}
\put(21.7,13.8){\line(0,1) {4.2}}
\put(17.2,10.3){$\beta_0^0 + \beta_1^0$}
\end{overpic} 

\notes{The data are obtained from 60 simulated time series with 20,000 observations each. The violin plots (black and red bordered areas) depict the state-dependent distributions of the dynamic causal effects computed as in Proposition \ref{prop:stvar_effects}. The dots and vertical lines are the group-wise means and standard deviations, respectively. The connected lines are the average state-dependent LP estimates from regression \eqref{eq:lp_bin}.}
\end{figure}  
  
\paragraph*{Causal Effects} I am interested in the effect of a government spending shock $X_t$ on GDP $Y_{t+h}$. For the STVAR model specified here, the structural function $\psi_h$ is differentiable. Thus, the causal effects of a marginal shock on $Y_{t+h}$, $\psi_h'(X_t, U_{h,t+h})$, always exist. From now on, suppress the dependence on $U_{h,t+h}$ by writing $\psi_h'(X_t)$. The derivative of the structural function is given by the following proposition.\footnote{The statement uses standard notation: for an $n \times n$ matrix $\boldsymbol{A}$, $\operatorname{chol}(\boldsymbol{A})$ is its lower-triangular Cholesky factor. The $n$-dimensional vector $\boldsymbol{\iota}_j$ is $(0,\ldots,0,1,0,\ldots,0)'$ with 1 at the $j$th entry. The elimination matrix $\boldsymbol{L}_n$ satisfies $\operatorname{vech}(\boldsymbol{A}) = \boldsymbol{L}_n \operatorname{vec}(\boldsymbol{A})$ for every $n \times n$ matrix $\boldsymbol{A}$ and $\operatorname{vec}(\boldsymbol{A}) = \boldsymbol{L}_n' \operatorname{vech}(\boldsymbol{A})$ for every lower-triangular $n \times n$ matrix $\boldsymbol{A}$. The commutation matrix $\boldsymbol{K}_{nn}$ satisfies $\operatorname{vec}(\boldsymbol{A}') = \boldsymbol{K}_{nn} \operatorname{vec}(\boldsymbol{A})$ for any $n \times n$ matrix $\boldsymbol{A}$.}

\begin{prop}\label{prop:stvar_effects}
    Suppose the endogenous vector $\boldsymbol{Y}_{t}$ satisfies the STVAR model defined by \eqref{eq:stvar_first}--\eqref{eq:stvar_last} with $S_t$ being the normalized average of $Y_t,\ldots,Y_{t-6}$, where $Y_t$ is ordered $r$th in $\boldsymbol{Y}_t$. Then the derivative $\boldsymbol{\psi}_h'(X_t)$ of the structural function of the vector $\boldsymbol{Y}_{t+h}$ with respect to $X_t \equiv \varepsilon_{1,t}$ for $h \geq 1$ is
    \begin{multline} \label{eq:stvar_effects}
    \boldsymbol{\psi}_h'(X_t) = \sum_{k=1}^p \left[ (1-F(S_{t+h-1}))\boldsymbol{\Pi}_{k,E} + F(S_{t+h-1})\boldsymbol{\Pi}_{k,R} \right] \boldsymbol{\psi}_{h-k}'(X_t) 
    \\+ \left[\frac{\partial F(S_{t+h-1})}{\partial X_t} \right] \left( \boldsymbol{c}_{R} - \boldsymbol{c}_{E} + \sum_{k=1}^p [\boldsymbol{\Pi}_{k,R} - \boldsymbol{\Pi}_{k,E}] \boldsymbol{Y}_{t+h-k} \right) \\ 
    + (\boldsymbol{\varepsilon}_{t+h}' \otimes \boldsymbol{I}_n) \boldsymbol{L}_n' (\boldsymbol{L}_n(\boldsymbol{I}_{n^2} + \boldsymbol{K}_{nn})(\boldsymbol{A}(S_{t+h-1}) \otimes \boldsymbol{I}_n)\boldsymbol{L}_n')^{-1} \\ \times \operatorname{vech}(\boldsymbol{\Omega}_R - \boldsymbol{\Omega}_E) \left[\frac{\partial F(S_{t+h-1})}{\partial X_t} \right],
    \end{multline}
    for $h = 0$ it is
    \begin{equation} \label{eq:stvar_eff2}
    \boldsymbol{\psi}_0'(X_t) = \boldsymbol{A}(S_{t-1}) \boldsymbol{\iota}_1,
    \end{equation}
    and for $h < 0$, $\boldsymbol{\psi}_h'(X_t) = \mathbf{0} \in \mathbb{R}^n$. The derivative of the state indicator with respect to $X_t$ is
    \begin{equation} \label{eq:stvar_eff3}
    \frac{\partial F(S_{t+h-1})}{\partial X_t} = F(S_{t+h-1})(1-F(S_{t+h-1})) \frac{- \gamma}{7 \sigma_s} \boldsymbol{\iota}_r' \left[ \sum_{k=1}^7 \boldsymbol{\psi}'_{h-k}(X_t) \right],
    \end{equation}
    where $\sigma_s^2$ is the variance of $\frac{1}{7}\sum_{k=1}^7 Y_{t+1-k}$.
    \end{prop}

The desired marginal effect $\psi'_h(X_t)$ for $Y_{t+h}$ is the third element of $\boldsymbol{\psi}_h'(X_t)$, which can be computed from Proposition \ref{prop:stvar_effects}. Note that the derivative of the conditional average structural function usually considered in the main paper is the conditional expectation of this `more granular' effect: 
\begin{equation*}
\Psi_h'(x,s) = \E[\psi_h'(x,U_{h,t+h}) \mid S_{t-1} = s].
\end{equation*}

\paragraph*{Simulation Exercise} To numerically verify Proposition 2, I proceed as follows: using the point estimate of the STVAR model as parameters, I simulate 60 time series with 20,000 observations each. Using Proposition \ref{prop:stvar_effects}, I calculate the marginal effect $\psi_h'(X_t)$ for every data point and for each time series I estimate the state-dependent LP
\begin{equation}\label{eq:lp_bin}
Y_{t+h} = X_t \beta_0^h +  \I[F(S_{t-1}) > 0.8] X_t \beta_1^h + \text{error}_{h,t+h}.
\end{equation}

Figure \ref{fig:effs_LP} displays the results of this exercise. The violin plots depict the distribution of the effects $\psi'_h(X_t)$ in the two states while the connected lines are the state-dependent LP estimates. The LP estimates correspond to the averages of the marginal effect distributions. This is as expected given Proposition 2: since $X_t$ follows a normal distribution, $\omega_X$ is the shock density and the average conditional effect $\theta_h(s,\omega_X)$ corresponds to the population conditional effect. As a next exercise, I estimate the state-dependent LP with continuous indicator
\begin{equation}\label{eq:lp_cont_stvar}
Y_{t+h} = X_t \beta_0^h + F(S_{t-1}) X_t \beta_1^h + \text{error}_{h,t+h}.
\end{equation}
  
\begin{figure}[t]
\justify
\caption{State-Dependent Distribution of Dynamic Shock Effects and LP Estimate.}\label{fig:prop_lp_cont}
\begin{minipage}{0.3\textwidth}
\centering
\begin{overpic}[scale=0.45,unit=1mm]{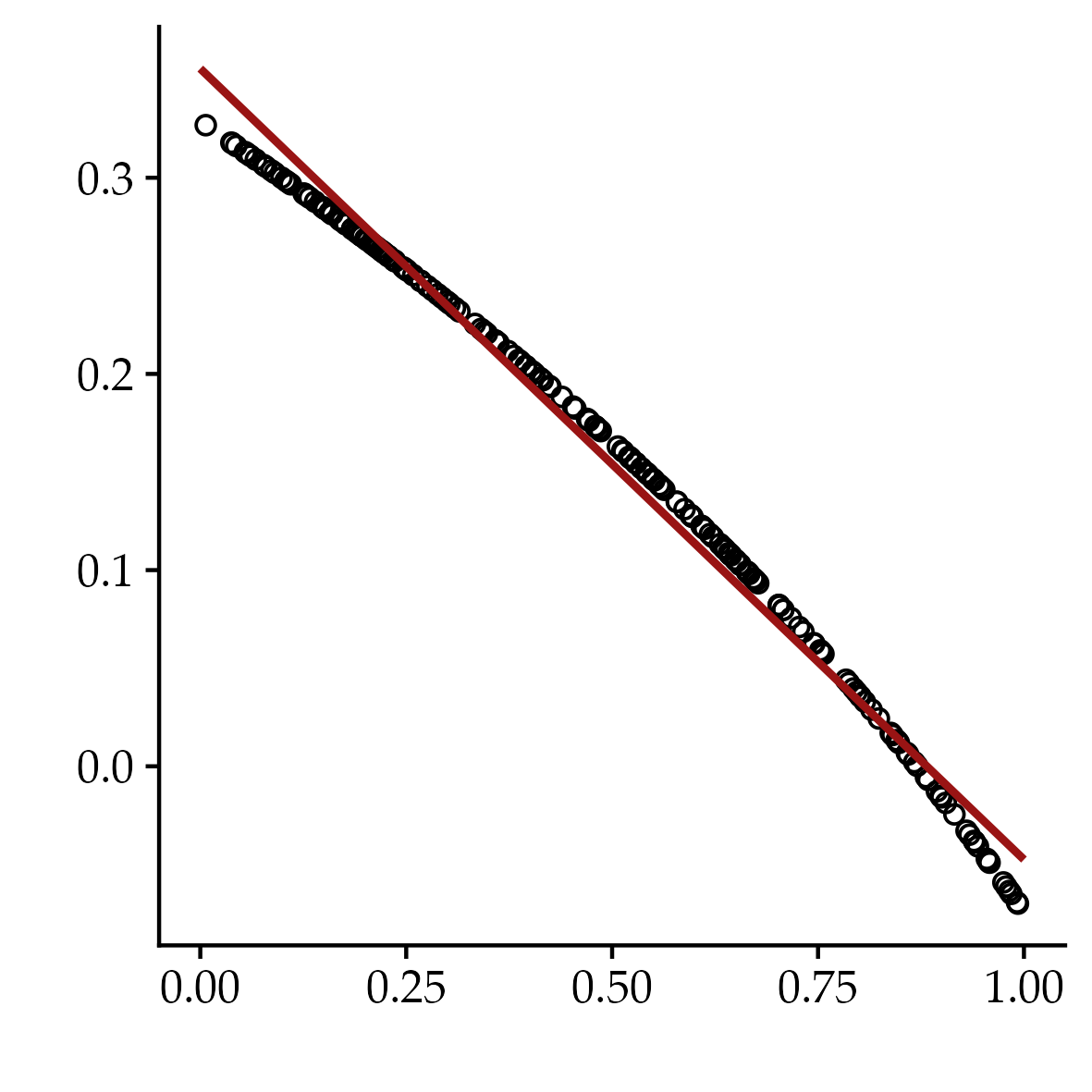}
\put(46,0.1){$F(S_{t-1})$}
\put(29,86){\color{figred}{\line(1,0){6}}}
\put(38,84){$\color{figred}{\beta_0^0 + F(S_{t-1}) \beta_1^0}$}
\put(60,60){\line(1,0){4}}
\put(65,58){$\psi_0'(X_t)$}
\end{overpic}
\subcaption{$h = 0$}
\end{minipage}
\hfill
\begin{minipage}{0.3\textwidth}
\centering
\begin{overpic}[scale=0.45,unit=1mm]{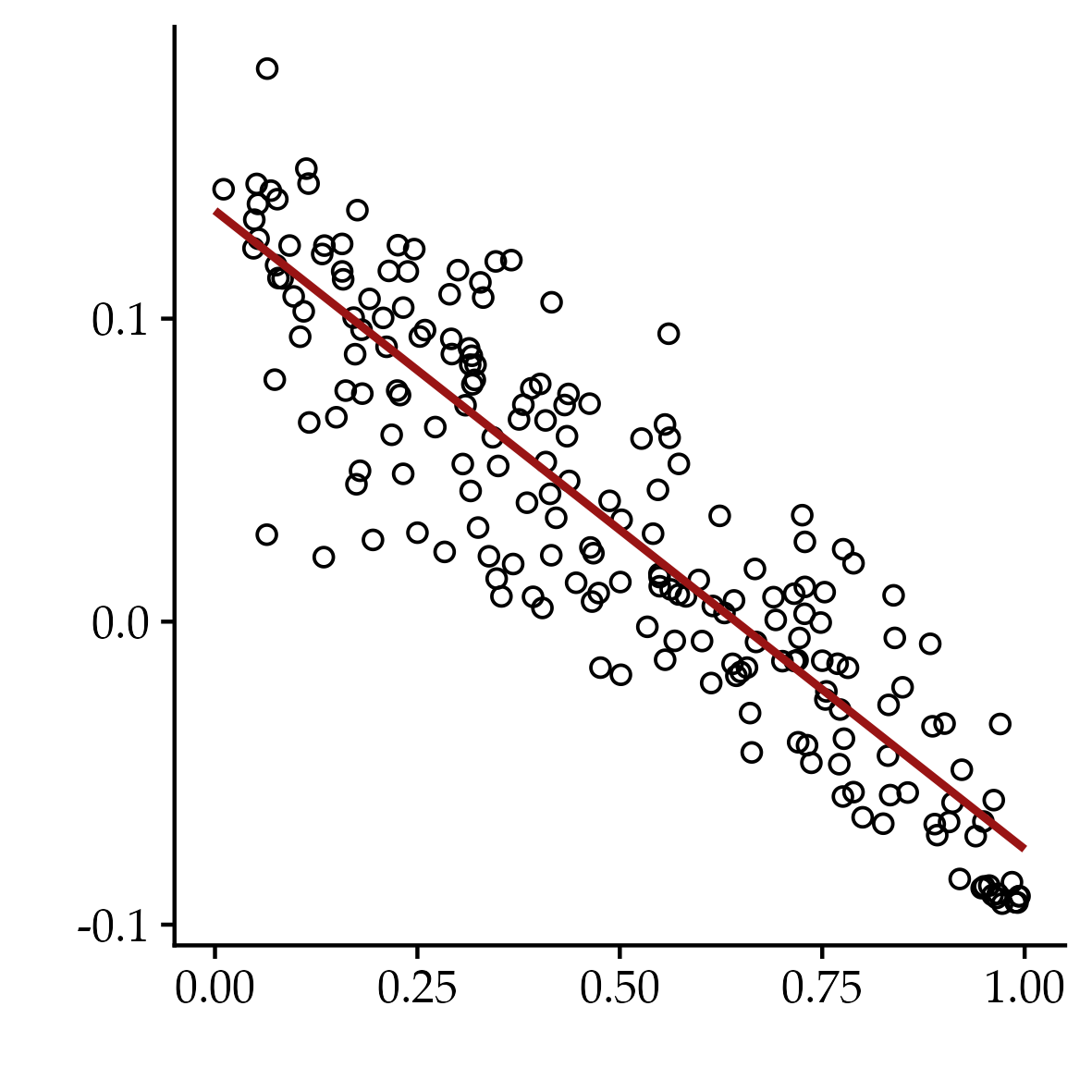}
\put(46,0.1){$F(S_{t-1})$}
\end{overpic}
\subcaption{$h=2$}
\end{minipage}
\hfill
\begin{minipage}{0.3\textwidth}
\centering
\begin{overpic}[scale=0.45,unit=1mm]{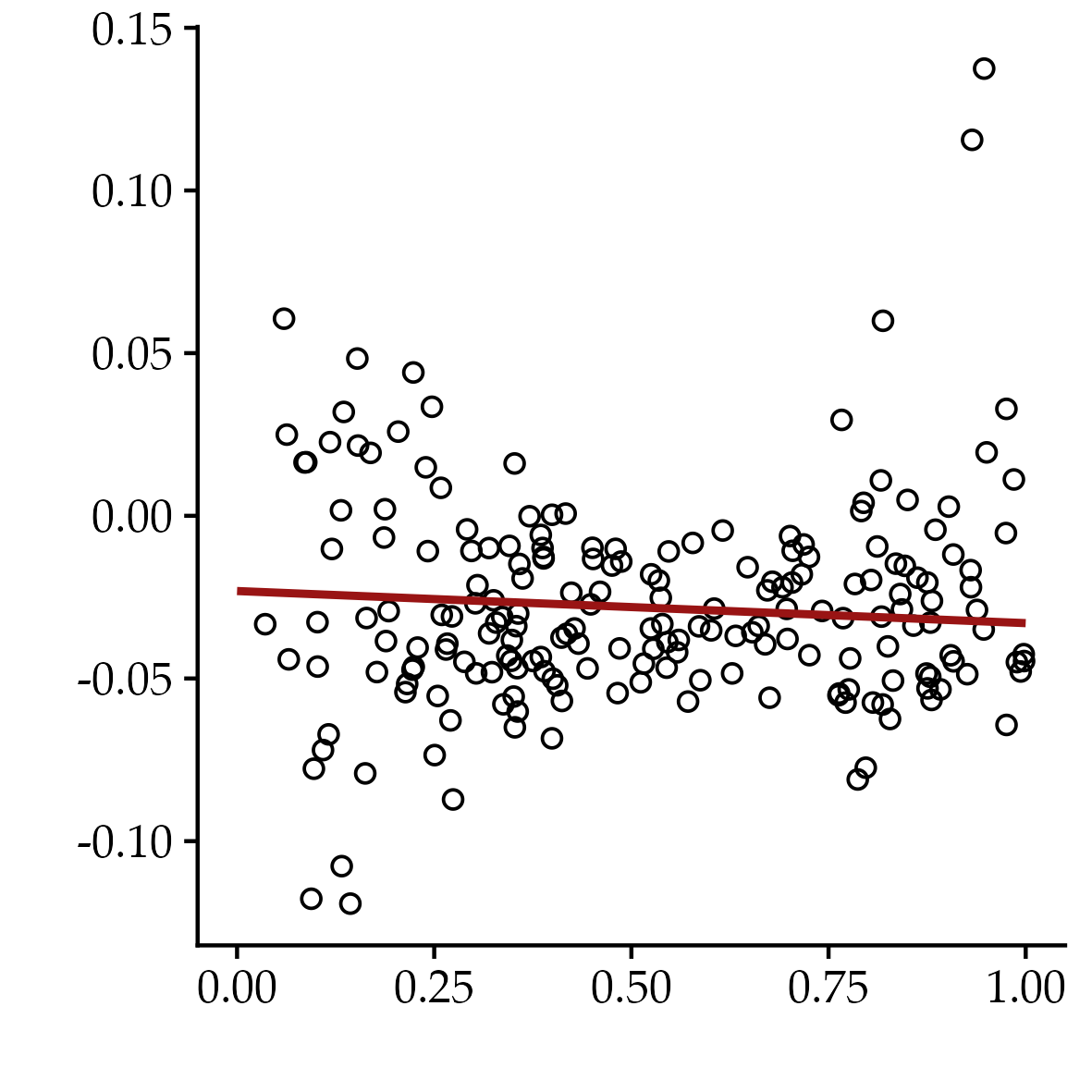}
\put(46,0.1){$F(S_{t-1})$}
\end{overpic}
\subcaption{$h=4$}
\end{minipage} \\
\notes{The red line is $\beta_0^h + F(S_{t-1}) \beta_1^h$ as obtained from the regression \eqref{eq:lp_cont_stvar}. The black circles are 200 randomly sampled causal effects from the STVAR model that are computed using Proposition \ref{prop:stvar_effects}.}
\end{figure}    
  
Figure \ref{fig:prop_lp_cont} compares the LP estimand of the conditional effect, $\beta_0^h + F(s) \beta_1^h$, to the actual marginal effects for three different horizons. The panels look like plots from a regression of the causal quantity $\psi_h'(X_t)$---which is unobserved in practice---onto $F(S_{t-1})$. However, the coefficients are the estimands of the state-dependent LP \eqref{eq:lp_cont_stvar}. This both confirms and visualizes the main point of Proposition 2.

\begin{proof}[Proof of Proposition \ref{prop:stvar_effects}]
    By the product rule, $\boldsymbol{\psi}_h'(X_t)$ follows the recursive formula
    \begin{align*}
    \boldsymbol{\psi}_h'(X_t) = {}& \sum_{k=1}^p \boldsymbol{\Pi}_{k}(S_{t+h-1}) \boldsymbol{\psi}_{h-k}'(X_t) + \frac{\partial \boldsymbol{c}(S_{t+h-1})}{\partial X_t} + \sum_{k=1}^p \left[\frac{\partial}{\partial X_t} \boldsymbol{\Pi}_{k}(S_{t+h-1})\right] \boldsymbol{Y}_{t+h-k} \\
    & + \left[ \frac{\partial}{\partial X_t} \boldsymbol{A}(S_{t+h-1}) \right] \boldsymbol{\varepsilon}_{t+h}.
    \end{align*}

    The first two lines of \eqref{eq:stvar_effects} are a direct application of this formula, since $\partial \boldsymbol{c}(S_{t+h-1})/\partial X_t = (\boldsymbol{c}_{R} - \boldsymbol{c}_{E}) \, \partial F(S_{t+h-1})/\partial X_t$ and analogously for $\boldsymbol{\Pi}_k(S_{t+h-1})$. For the last term, note that
    \begin{align*}
        \boldsymbol{A}(S_{t+h-1}) \boldsymbol{\varepsilon}_{t+h} & = \operatorname{vec}(\boldsymbol{A}(S_{t+h-1}) \boldsymbol{\varepsilon}_{t+h}) \\
    & = (\boldsymbol{\varepsilon}_{t+h}' \otimes \boldsymbol{I}_n) \operatorname{vec}(\boldsymbol{A}(S_{t+h-1})),
    \end{align*}
    where the second equality follows from \citesupp[p.~668, (6)]{Lutkepohl:05}. Now define $\boldsymbol{\Omega}_{t+h} \equiv (1-F(S_{t+h-1}))\boldsymbol{\Omega}_E + F(S_{t+h-1})\boldsymbol{\Omega}_R$ and apply the chain rule multiple times:
    \begin{multline*}
    \frac{\partial\operatorname{vec}( \boldsymbol{A}(S_{t+h-1}))}{\partial X_t} = \overbrace{\frac{\partial\operatorname{vec}( \operatorname{chol}(\boldsymbol{\Omega}_{t+h}))}{\partial \operatorname{vech}(\operatorname{chol}(\boldsymbol{\Omega}_{t+h}))}}^{I \equiv } \overbrace{\frac{\partial \operatorname{vech}(\operatorname{chol}(\boldsymbol{\Omega}_{t+h}))}{\partial\operatorname{vech}(\boldsymbol{\Omega}_{t+h})}}^{\mathit{II} \equiv } \\
    \overbrace{\frac{\partial\operatorname{vech}(\boldsymbol{\Omega}_{t+h})}{\partial F(S_{t+h-1})}}^{\mathit{III} \equiv } \frac{\partial F(S_{t+h-1})}{\partial X_t}.
    \end{multline*}
    Since $\operatorname{chol}(\boldsymbol{\Omega}_{t+h})$ is lower triangular, $\operatorname{vec}(\operatorname{chol}(\boldsymbol{\Omega}_{t+h})) = \boldsymbol{L}_n' \operatorname{vech}(\operatorname{chol}(\boldsymbol{\Omega}_{t+h}))$, so
    \begin{equation*}
    I = \boldsymbol{L}_n'.
    \end{equation*}
    For $\mathit{II}$, use \citesupp[p.~669, (10)]{Lutkepohl:05}, which yields:
    \begin{equation*}
    \mathit{II} = (\boldsymbol{L}_n(\boldsymbol{I}_{n^2} + \boldsymbol{K}_{nn})(\boldsymbol{A}(S_{t+h-1}) \otimes \boldsymbol{I}_n)\boldsymbol{L}_n')^{-1}.
    \end{equation*}
    For $\mathit{III}$, due to linearity of the $\operatorname{vech}$ operator,
    \begin{equation*}
    \operatorname{vech}(\boldsymbol{\Omega}_{t+h}) = \operatorname{vech}(\boldsymbol{\Omega}_E) + F(S_{t+h-1}) \operatorname{vech}(\boldsymbol{\Omega}_R - \boldsymbol{\Omega}_E).
    \end{equation*}
    It follows immediately that
    \begin{equation*}
    \mathit{III} = \operatorname{vech}(\boldsymbol{\Omega}_R - \boldsymbol{\Omega}_E).
    \end{equation*}
    This finishes the proof of equation \eqref{eq:stvar_effects}. For equations \eqref{eq:stvar_eff2} and \eqref{eq:stvar_eff3} recall the derivative properties of the logistic function and note that left-multiplying a matrix by $\boldsymbol{\iota}_r'$ selects the $r$th row and right-multiplying by $\boldsymbol{\iota}_r$ selects the $r$th column.
\end{proof}

\section{Details of the Simulation Study}\label{app:simdetails}

\paragraph*{Data} The calibration uses the quarterly macroeconomic panel of \citesupp{Stock:16} over 1959Q3--2014Q4 (222 quarters). Only series observed in every quarter of this sample are kept, which leaves a balanced panel of 145 series. Following them, nominal series are deflated, monthly series are aggregated to quarters, and all series are transformed to stationarity. Outliers (defined as observations whose deviation from the sample median exceeds 4.5 or 3 times the interquartile range) are replaced by the median of the observation and the five preceding ones, and a local biweight mean (bandwidth 100 quarters) is removed. The factors $\boldsymbol{f}_t$ are the first six principal components of the standardized balanced panel of the 94 series used for factor estimation by \citesupp{Stock:16}.

\paragraph*{States} The unemployment state indicates that the average unemployment rate of adult men and women is in its upper 30 percent; the growth state indicates that the seven-quarter centered moving average of real GDP growth is in its lower 30 percent. Each state is treated as a Markov chain with transition probabilities equal to the sample frequencies. The probabilities of staying in the low and high states are 0.96 and 0.91 for unemployment and 0.93 and 0.82 for growth, with mean spells in the slack state of 11.3 and 5.4 quarters.

\paragraph*{Factor dynamics} For each state variable, the coefficients $\boldsymbol{\Gamma}_\ell(s)$ of
\begin{equation*}
   \boldsymbol{f}_t = \sum_{\ell = 1}^4 \boldsymbol{\Gamma}_\ell(S_{t-\ell}) \boldsymbol{f}_{t-\ell} + \boldsymbol{H}(S_t) \boldsymbol{\varepsilon}_t, \qquad \boldsymbol{\varepsilon}_t \sim N(\mathbf{0}, \boldsymbol{I}_6),
\end{equation*}
are estimated by least squares, interacting each lag with its own state. The impact matrix is $\boldsymbol{H}(s) = \boldsymbol{C}(s)\boldsymbol{Q}(s)$, where $\boldsymbol{C}(s)$ is the Cholesky factor of the residual covariance matrix in state $s$ and $\boldsymbol{Q}(s)$ is orthonormal with first column proportional to $\boldsymbol{C}(s)'\boldsymbol{\lambda}_G$, where $\boldsymbol{\lambda}_G$ is the loading of real federal government consumption. The first shock, $X_t = \varepsilon_{1,t}$, therefore has the largest impact on government spending among all unit-variance combinations of $\boldsymbol{\varepsilon}_t$. The observed series are $\boldsymbol{\lambda}_i'\boldsymbol{f}_t + u_{i,t}$, with state-invariant loadings $\boldsymbol{\lambda}_i$ and idiosyncratic AR(4) errors $u_{i,t}$ with Gaussian innovations, all estimated by least squares.

\paragraph*{DGPs} For each state variable there are 100 DGPs. Each DGP consists of the shock $X_t$ and five series: real federal government consumption, one real-activity series, one price series, one labor-market series (the outcome) and one further series drawn at random from the panel. Each sample has $T=240$ quarters after a burn-in of 200; state paths in which either state makes up less than 20 percent of the sample are redrawn.

\paragraph*{Estimators} All estimators use $p=4$ lags of the six series and a constant, and horizons $h=0,\ldots,16$. The LPs and the fixed-state, moving-state and linear VARs use all periods $t \geq p+1$; the moving-state VAR uses the true transition probabilities. The backshifted VAR uses the common sample $t \geq H+2$ for all its regressions. The impact of all VARs is the coefficient on $X_t$ in the regression of the observed series on $X_t$ and the lags, i.e.\ the response to a unit shock.

\paragraph*{Targets} The true conditional impulse response of the outcome to a unit shock $X_t$ given $S_{t-1}=s$ is computed exactly by averaging the factor responses over all future paths of the Markov state. The unconditional target averages the two conditional responses with the stationary state probabilities.

\paragraph*{Inference} Confidence intervals are $\hat\theta \pm 1.645 \cdot \text{s.e.}$, with heteroskedasticity-robust (HC1) standard errors for the LP and delta-method standard errors for the backshifted and linear VAR (Section \ref{app:delta}). Figure \ref{fig:sim_seacc} reports the ratio of the median standard error to the standard deviation of the estimates across samples.

\begin{figure}[H]
  \centering
  \caption{Accuracy of the Standard Errors.}\label{fig:sim_seacc}
  \begin{subfigure}[b]{0.42\textwidth}
    \centering
    \includegraphics[width=\linewidth]{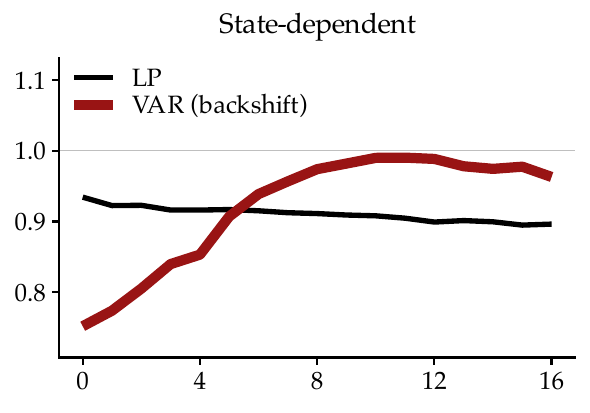}
  \end{subfigure}\hspace{0.06\textwidth}
  \begin{subfigure}[b]{0.42\textwidth}
    \centering
    \includegraphics[width=\linewidth]{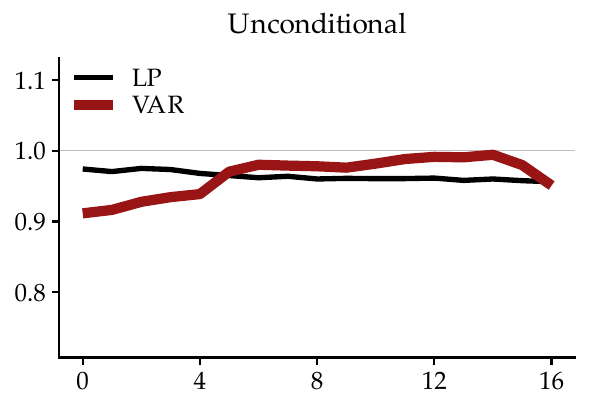}
  \end{subfigure}
  \notes{Median standard error divided by the standard deviation of the estimates across samples, median over DGPs (and both states in the left panel). Standard errors: HC1 for the LP; delta method for the VAR. Left: state-dependent LP and backshifted VAR, both states. Right: linear LP and VAR. 100 DGPs per state variable, 1,000 samples of length $T=240$ each, $p=4$ lags.}
\end{figure}

\section{Delta-Method Standard Errors for the Backshifted VAR}\label{app:delta}

The simulations in Section 5 and the empirical application in Section 6 of the main paper use analytical delta-method standard errors for inference on the (backshifted) VAR impulse responses. This section shows in detail how these standard errors are constructed.

Let $n$ denote the number of variables in $\boldsymbol{Y}_t$ and $p$ the number of lags. In companion form, $\boldsymbol{Y}_{t-1}$ stacks the $p$ lags of the $n$ variables, $\boldsymbol{\Pi}^k(s)$ is the $np \times np$ companion matrix and $\boldsymbol{A}^k(s)$ is $np \times n$.
   
Denote the vector of estimated projection coefficients that are necessary to construct impulse responses as 
\begin{equation*}
     \boldsymbol{\pi}(s) \equiv \begin{pmatrix}
         \boldsymbol{A}_{2:n,1}^{0 \prime}(s) & \operatorname{vec}(\boldsymbol{\Pi}^1_{1:n,\bullet}(s))' & \ldots  & \operatorname{vec}(\boldsymbol{\Pi}^h_{1:n,\bullet}(s))'
     \end{pmatrix}'.
\end{equation*}

Note that the (1,1) entry of $\boldsymbol{A}^0(s)$ is always one, so it can be omitted, and due to the companion form only the first $n$ rows of $\boldsymbol{\Pi}^k(s)$ have to be estimated. The dimension of the vector $\boldsymbol{\pi}(s)$ is $(n-1) + h n^2 p$.
  
To estimate the projection coefficients, one can just run the regressions
\begin{align*}
    \boldsymbol{Y}_{2:n,t} & = \boldsymbol{A}^0_{2:n,1}(S_{t-1}) X_t + \boldsymbol{\Pi}^0_{2:n,\bullet}(S_{t-1}) \boldsymbol{Y}_{t-1} + \boldsymbol{A}^0_{2:n,2:n}(S_{t-1}) (\boldsymbol{E}_t^{0,\perp})_{2:n},\\
    \boldsymbol{Y}_{1:n,t} & = \boldsymbol{\Pi}^1_{1:n,\bullet}(S_{t-2}) \boldsymbol{Y}_{t-1} + \boldsymbol{A}^1_{1:n,\bullet}(S_{t-2}) \boldsymbol{E}_t^{1,\perp}, \\
    \vdots & \qquad \qquad \vdots \\
    \boldsymbol{Y}_{1:n,t} & = \boldsymbol{\Pi}^h_{1:n,\bullet}(S_{t-1-h}) \boldsymbol{Y}_{t-1} + \boldsymbol{A}^h_{1:n,\bullet}(S_{t-1-h}) \boldsymbol{E}_t^{h,\perp}.
\end{align*}

The variance--covariance matrix $\widehat{\boldsymbol{V}}(s)$ of $\widehat{\boldsymbol{\pi}}(s)$ can then be constructed using standard GMM formulas. The main paper uses a heteroskedasticity and autocorrelation robust variance estimator for the `meat' of the variance.
  
Recall the definition of the backshifted VAR:
\begin{align*}
    \theta^b_{\mathit{VAR,h}}(s) \equiv \left(\boldsymbol{\Pi}^h(s) \cdot \ldots \cdot \boldsymbol{\Pi}^1(s) \boldsymbol{A}^0(s) \right)_{2,1}.
\end{align*}

To apply the delta method, note that for any $k$ between $1$ and $h$,
\begin{multline*}
  \boldsymbol{\Pi}^h(s) \cdot \ldots \cdot \boldsymbol{\Pi}^1(s) \boldsymbol{A}^0_{\bullet,1}(s) \\ = \left(\left( \boldsymbol{\Pi}^{k-1}(s) \cdot \ldots \cdot \boldsymbol{\Pi}^{1}(s) \boldsymbol{A}^0_{\bullet,1}(s) \right)' \otimes \left(\boldsymbol{\Pi}^{h}(s) \cdot \ldots \cdot \boldsymbol{\Pi}^{k+1}(s) \right) \right) \operatorname{vec}(\boldsymbol{\Pi}^k(s)).
\end{multline*}

Since only the first $n$ rows of $\boldsymbol{\Pi}^k(s)$ are estimated, only the first $n$ columns of $\boldsymbol{\Pi}^h(s) \cdot \ldots \cdot \boldsymbol{\Pi}^{k+1}(s)$ enter the derivative. It follows that the relevant gradient is
\begin{equation*}
   \boldsymbol{G} \equiv \frac{\partial \theta^b_{\mathit{VAR,h}}(s)}{\partial {\boldsymbol{\pi}}(s)'} = \begin{pmatrix}
       \left( \boldsymbol{\Pi}^h(s) \cdot \ldots \cdot \boldsymbol{\Pi}^1(s)  \right)_{2,2:n}' \\
       \left(\boldsymbol{A}^{0 \prime}_{\bullet,1}(s) \otimes \left( \boldsymbol{\Pi}^h(s) \cdot \ldots \cdot \boldsymbol{\Pi}^2(s) \right)_{2,1:n} \right)' \\
       \vdots \\
       \left((\boldsymbol{\Pi}^{h-1}(s) \cdot \ldots \cdot \boldsymbol{\Pi}^1(s)  \boldsymbol{A}^0_{\bullet,1}(s) )' \otimes \boldsymbol{I}_{2,1:n} \right)'
   \end{pmatrix}'.
\end{equation*}

Using a plug-in estimator $\widehat{\boldsymbol{G}}$ for the gradient $\boldsymbol{G}$, the delta-method standard errors are computed as 
\begin{equation*}
    \sqrt{\widehat{\boldsymbol{G}} \widehat{\boldsymbol{V}} \widehat{\boldsymbol{G}}'}.
\end{equation*}

\begingroup
\expandafter\let\csname @auxout\expandafter\endcsname\csname @auxoutsupp\endcsname

\endgroup


\clearpage

\begin{thebibliography}{39}
\newcommand{\enquote}[1]{``#1''}
\expandafter\ifx\csname natexlab\endcsname\relax\def\natexlab#1{#1}\fi

\bibitem[\protect\citeauthoryear{Almuzara and Sancibri{\'a}n}{Almuzara and
  Sancibri{\'a}n}{2026}]{Almuzara:26}
\textsc{Almuzara, Mart{\'\i}n and V{\'\i}ctor Sancibri{\'a}n} (2026):
  \enquote{Micro Responses to Macro Shocks,} Federal Reserve Bank of New York
  Staff Report No. 1090.

\bibitem[\protect\citeauthoryear{Angrist and Kuersteiner}{Angrist and
  Kuersteiner}{2011}]{Angrist:11}
\textsc{Angrist, Joshua and Guido~M. Kuersteiner} (2011): \enquote{Causal
  Effects of Monetary Shocks: Semiparametric Conditional Independence Tests
  with a Multinomial Propensity Score,} \emph{Review of Economics and
  Statistics}, 93 (3), 725--747.

\bibitem[\protect\citeauthoryear{Angrist, Graddy, and Imbens}{Angrist
  et~al.}{2000}]{Angrist:00}
\textsc{Angrist, Joshua~D., Kathryn Graddy, and Guido~W. Imbens} (2000):
  \enquote{The Interpretation of Instrumental Variables Estimators in
  Simultaneous Equation Models with an Application to the Demand for Fish,}
  \emph{Review of Economic Studies}, 67, 499--527.

\bibitem[\protect\citeauthoryear{Angrist, Jord{\`a}, and Kuersteiner}{Angrist
  et~al.}{2018}]{Angrist:18}
\textsc{Angrist, Joshua~D., {\`O}scar Jord{\`a}, and Guido~M. Kuersteiner}
  (2018): \enquote{Semiparametric Estimates of Monetary Policy Effects: String
  Theory Revisited,} \emph{Journal of Business and Economic Statistics}, 36
  (3), 371--387.

\bibitem[\protect\citeauthoryear{Aruoba, Mlikota, Schorfheide, and
  Villalvazo}{Aruoba et~al.}{2022}]{Aruoba:22}
\textsc{Aruoba, S.~Bora{\u g}an, Marko Mlikota, Frank Schorfheide, and Sergio
  Villalvazo} (2022): \enquote{{SVARs} with Occasionally-Binding Constraints,}
  \emph{Journal of Econometrics}, 231, 477--499.

\bibitem[\protect\citeauthoryear{Ascari and Mavroeidis}{Ascari and
  Mavroeidis}{2022}]{Ascari:22}
\textsc{Ascari, Guido and Sophocles Mavroeidis} (2022): \enquote{The Unbearable
  Lightness of Equilibria in a Low Interest Rate Environment,} \emph{Journal of
  Monetary Economics}, 127, 1--17.

\bibitem[\protect\citeauthoryear{Auer, Bernardini, and Cecioni}{Auer
  et~al.}{2021}]{Auer:21}
\textsc{Auer, Simone, Marco Bernardini, and Martina Cecioni} (2021):
  \enquote{Corporate leverage and monetary policy effectiveness in the euro
  are,} \emph{European Economic Review}, 140, 103943.

\bibitem[\protect\citeauthoryear{Auerbach and Gorodnichenko}{Auerbach and
  Gorodnichenko}{2012}]{Auerbach:12}
\textsc{Auerbach, Alan~J. and Yuriy Gorodnichenko} (2012): \enquote{Measuring
  the Output Responses to Fiscal Policy,} \emph{American Economic Journal:
  Economic Policy}, 4 (2), 1--27.

\bibitem[\protect\citeauthoryear{Auerbach and Gorodnichenko}{Auerbach and
  Gorodnichenko}{2013}]{Auerbach:13}
---\hspace{-.1pt}---\hspace{-.1pt}--- (2013): \enquote{Fiscal Multipliers in
  Recession and Expansion,} in \emph{Fiscal Policy after the Financial Crisis},
  ed. by Alberto Alesina and Francesco Giavazzi, University of Chicago Press,
  63--98.

\bibitem[\protect\citeauthoryear{Blundell and Powell}{Blundell and
  Powell}{2003}]{Blundell:03}
\textsc{Blundell, Richard and James~L. Powell} (2003): \enquote{Endogeneity in
  Nonparametric and Semiparametric Regression Models,} in \emph{Advances in
  Economics and Econometrics: Theory and Applications, Eighth World Congress},
  ed. by Mathias Dewatripont, Lars~Peter Hansen, and Stephen~J. Turnovsky,
  Cambridge University Press, vol.~2, 312--357.

\bibitem[\protect\citeauthoryear{Bousquet}{Bousquet}{2026}]{Bousquet:26}
\textsc{Bousquet, Paul} (2026): \enquote{Minimal Assumptions to Uncover Macro
  Dynamics Using ``Shocks'',} Working Paper, University of Virginia.

\bibitem[\protect\citeauthoryear{Caravello and Bruera}{Caravello and
  Bruera}{2024}]{Caravello:24}
\textsc{Caravello, Tom{\'a}s~E. and Pedro~Mart{\'\i}nez Bruera} (2024):
  \enquote{Disentangling Sign and Size Non-Linearities,} Working Paper.

\bibitem[\protect\citeauthoryear{Casini and McCloskey}{Casini and
  McCloskey}{2025}]{Casini:25}
\textsc{Casini, Alessandro and Adam McCloskey} (2025): \enquote{Identification,
  Estimation and Inference in High-Frequency Event Study Regressions,} CEIS
  Research Paper No. 608.

\bibitem[\protect\citeauthoryear{Casini, McCloskey, Rolla, and Pala}{Casini
  et~al.}{2026}]{Casini:26}
\textsc{Casini, Alessandro, Adam McCloskey, Luca Rolla, and Raimondo Pala}
  (2026): \enquote{Dynamic Local Average Treatment Effects in Time Series,}
  ArXiv:2509.12985.

\bibitem[\protect\citeauthoryear{Cloyne, Jord{\`a}, and Taylor}{Cloyne
  et~al.}{2023}]{Cloyne:23}
\textsc{Cloyne, James, {\`O}scar Jord{\`a}, and Alan~M. Taylor} (2023):
  \enquote{State-Dependent Local Projections: Understanding Impulse Response
  Heterogeneity,} NBER Working Paper No. 30971.

\bibitem[\protect\citeauthoryear{David, Giacomini, Jiao, and Wang}{David
  et~al.}{2026}]{David:26}
\textsc{David, Joel~M., Raffaella Giacomini, Xiyu Jiao, and Weining Wang}
  (2026): \enquote{Causal State-Dependent Local Projections,} Federal Reserve
  Bank of Chicago Working Paper No. 2026-18.

\bibitem[\protect\citeauthoryear{Gon{\c c}alves, Herrera, Kilian, and
  Pesavento}{Gon{\c c}alves et~al.}{2024{\natexlab{a}}}]{Goncalves:24b}
\textsc{Gon{\c c}alves, S{\'\i}lvia, Ana~Mar{\'\i}a Herrera, Lutz Kilian, and
  Elena Pesavento} (2024{\natexlab{a}}): \enquote{Nonparametric Local
  Projections,} Federal Reserve Bank of Dallas Working Paper No. 2414.

\bibitem[\protect\citeauthoryear{Gon{\c c}alves, Herrera, Kilian, and
  Pesavento}{Gon{\c c}alves et~al.}{2024{\natexlab{b}}}]{Goncalves:24}
---\hspace{-.1pt}---\hspace{-.1pt}--- (2024{\natexlab{b}}):
  \enquote{State-dependent local projections,} \emph{Journal of Econometrics},
  244 (2), 105702.

\bibitem[\protect\citeauthoryear{Gon{\c c}alves, Herrera, Pesavento, Kilian,
  and {Kelanemer Holban}}{Gon{\c c}alves et~al.}{2026}]{Goncalves:26}
\textsc{Gon{\c c}alves, S{\'\i}lvia, Ana~Mar{\'\i}a Herrera, Elena Pesavento,
  Lutz Kilian, and Iones {Kelanemer Holban}} (2026): \enquote{Semiparametric
  Local Projections,} ArXiv:2606.13519.

\bibitem[\protect\citeauthoryear{Graham and Pinto}{Graham and
  Pinto}{2022}]{Graham:22}
\textsc{Graham, Bryan~S. and {Cristine Campos de Xavier} Pinto} (2022):
  \enquote{Semiparametrically efficient estimation of the average linear
  regression function,} \emph{Journal of Econometrics}, 226 (1), 115--138.

\bibitem[\protect\citeauthoryear{Hamilton}{Hamilton}{1989}]{Hamilton:89}
\textsc{Hamilton, James~D.} (1989): \enquote{A New Approach to the Economic
  Analysis of Nonstationary Time Series and the Business Cycle,}
  \emph{Econometrica}, 57 (2), 357--384.

\bibitem[\protect\citeauthoryear{Jord{\`a}}{Jord{\`a}}{2005}]{Jorda:05}
\textsc{Jord{\`a}, {\`O}scar} (2005): \enquote{Estimation and Inference of
  Impulse Responses by Local Projections,} \emph{American Economic Review}, 95
  (1), 161--182.

\bibitem[\protect\citeauthoryear{Jord{\`a}}{Jord{\`a}}{2026}]{Jorda:26}
---\hspace{-.1pt}---\hspace{-.1pt}--- (2026): \enquote{Comment: Local
  Projections or {VARs}? {A} Primer for Macroeconomists,} \emph{NBER
  Macroeconomics Annual}, 40 (1), 172--186.

\bibitem[\protect\citeauthoryear{Koles{\'a}r and
  Plagborg-M{\o}ller}{Koles{\'a}r and
  Plagborg-M{\o}ller}{2025{\natexlab{a}}}]{Kolesar:24}
\textsc{Koles{\'a}r, Michal and Mikkel Plagborg-M{\o}ller}
  (2025{\natexlab{a}}): \enquote{Dynamic Causal Effects in a Nonlinear World:
  the Good, the Bad, and the Ugly,} \emph{Journal of Business and Economic
  Statistics}, 43 (4), 737--754.

\bibitem[\protect\citeauthoryear{Koles{\'a}r and
  Plagborg-M{\o}ller}{Koles{\'a}r and
  Plagborg-M{\o}ller}{2025{\natexlab{b}}}]{Kolesar:25}
---\hspace{-.1pt}---\hspace{-.1pt}--- (2025{\natexlab{b}}):
  \enquote{Rejoinder,} \emph{Journal of Business and Economic Statistics}, 43
  (4), 778--781.

\bibitem[\protect\citeauthoryear{Li, Plagborg-M{\o}ller, and Wolf}{Li
  et~al.}{2024}]{Li:24}
\textsc{Li, Dake, Mikkel Plagborg-M{\o}ller, and Christian~K. Wolf} (2024):
  \enquote{Local Projections vs. {VARs}: Lessons from Thousands of {DGPs},}
  \emph{Journal of Econometrics}, 244 (2), 105722.

\bibitem[\protect\citeauthoryear{Ludwig}{Ludwig}{2026}]{Ludwig:24}
\textsc{Ludwig, Julian~F.} (2026): \enquote{Local Projections Are {VAR}
  Predictions of Increasing Order,} Working Paper.

\bibitem[\protect\citeauthoryear{Manuel and Wolf}{Manuel and
  Wolf}{2026}]{Manuel:26}
\textsc{Manuel, Ed and Christian~K. Wolf} (2026): \enquote{Identifying Policy
  Causal Effects from Rule Changes,} NBER Working Paper No. 35615.

\bibitem[\protect\citeauthoryear{Mavroeidis}{Mavroeidis}{2021}]{Mavroeidis:21}
\textsc{Mavroeidis, Sophocles} (2021): \enquote{Identification at the Zero
  Lower Bound,} \emph{Econometrica}, 89 (6), 2855--2885.

\bibitem[\protect\citeauthoryear{{Montiel~Olea}, Plagborg-M{\o}ller, Qian, and
  Wolf}{{Montiel~Olea} et~al.}{2026}]{Olea:24}
\textsc{{Montiel~Olea}, Jos{\'e}~Luis, Mikkel Plagborg-M{\o}ller, Eric Qian,
  and Christian~K. Wolf} (2026): \enquote{Double Robustness of Local
  Projections and Some Unpleasant {VARithmetic},} \emph{Econometrica}, 94 (4),
  1313--1343.

\bibitem[\protect\citeauthoryear{Paranhos}{Paranhos}{2025}]{Paranhos:25}
\textsc{Paranhos, Livia} (2025): \enquote{How do firms' financial conditions
  influence the transmission of monetary policy? A non-parametric local
  projection approach,} \emph{Journal of Econometrics}, 249, 105886.

\bibitem[\protect\citeauthoryear{Plagborg-M{\o}ller and
  Wolf}{Plagborg-M{\o}ller and Wolf}{2021}]{Plagborg-Moller:21}
\textsc{Plagborg-M{\o}ller, Mikkel and Christian~K. Wolf} (2021):
  \enquote{Local Projections and Vars Estimate the same Impulse Responses,}
  \emph{Econometrica}, 89 (2), 955--980.

\bibitem[\protect\citeauthoryear{Rambachan and Shephard}{Rambachan and
  Shephard}{2025}]{Rambachan:21}
\textsc{Rambachan, Ashesh and Neil Shephard} (2025): \enquote{When do common
  time series estimands have nonparametric causal meaning?} ArXiv:1903.01637.

\bibitem[\protect\citeauthoryear{Ramey and Zubairy}{Ramey and
  Zubairy}{2018}]{Ramey:18}
\textsc{Ramey, Valerie~A. and Sarah Zubairy} (2018): \enquote{Government
  Spending Multipliers in Good Times and in Bad: Evidence from US Historical
  Data,} \emph{Journal of Political Economy}, 126 (2), 850--901.

\bibitem[\protect\citeauthoryear{Stock and Watson}{Stock and
  Watson}{2016}]{Stock:16}
\textsc{Stock, James~H. and Mark~W. Watson} (2016): \enquote{Dynamic Factor
  Models, Factor-Augmented Vector Autoregressions, and Structural Vector
  Autoregressions in Macroeconomics,} in \emph{Handbook of Macroeconomics}, ed.
  by John~B. Taylor and Harald Uhlig, Elsevier, vol.~2A, 415--525.

\bibitem[\protect\citeauthoryear{White}{White}{1980}]{White:80}
\textsc{White, Halbert} (1980): \enquote{Using Least Squares to Approximate
  Unknown Regression Functions,} \emph{International Economic Review}, 21 (1),
  149--170.

\bibitem[\protect\citeauthoryear{Xu}{Xu}{2026}]{Xu:26}
\textsc{Xu, Ke-Li} (2026): \enquote{Local Projection-Based Inference under
  General Conditions,} \emph{Review of Economic Studies}, forthcoming.

\bibitem[\protect\citeauthoryear{Yitzhaki}{Yitzhaki}{1996}]{Yitzhaki:96}
\textsc{Yitzhaki, Shlomo} (1996): \enquote{On Using Linear Regressions in
  Welfare Economics,} \emph{Journal of Business and Economic Statistics}, 14
  (4), 478--486.

\bibitem[\protect\citeauthoryear{You}{You}{2026}]{You:26}
\textsc{You, Zhiheng} (2026): \enquote{How Well Are State-Dependent Local
  Projections Capturing Nonlinearities?} ArXiv:2602.14455.

\end{thebibliography}

\begin{thebibliography}{52}
\newcommand{\enquote}[1]{``#1''}
\expandafter\ifx\csname natexlab\endcsname\relax\def\natexlab#1{#1}\fi

\bibitem[\protect\citeauthoryear{Albrizio, Choi, Furceri, and Yoon}{Albrizio
  et~al.}{2020}]{Albrizio:20}
\textsc{Albrizio, Silvia, Sangyup Choi, Davide Furceri, and Chansik Yoon}
  (2020): \enquote{International Bank Lending Channel of Monetary Policy,}
  \emph{Journal of International Money and Finance}, 102, 102124.

\bibitem[\protect\citeauthoryear{Albuquerque}{Albuquerque}{2019}]{Albuquerque:19}
\textsc{Albuquerque, Bruno} (2019): \enquote{One Size Fits All? Monetary Policy
  and Asymmetric Household Debt Cycles in U.S. States,} \emph{Journal of Money,
  Credit and Banking}, 51 (5), 1309--1353.

\bibitem[\protect\citeauthoryear{Alloza}{Alloza}{2022}]{Alloza:22}
\textsc{Alloza, Mario} (2022): \enquote{Is Fiscal Policy More Effective During
  Recessions?} \emph{International Economic Review}, 63 (3), 1271--1292.

\bibitem[\protect\citeauthoryear{Alpanda, Granziera, and Zubairy}{Alpanda
  et~al.}{2021}]{Alpanda:21}
\textsc{Alpanda, Sami, Eleonora Granziera, and Sarah Zubairy} (2021):
  \enquote{State dependence of monetary policy across business, credit and
  interest rate cycles,} \emph{European Economic Review}, 140, 103936.

\bibitem[\protect\citeauthoryear{Auer, Bernardini, and Cecioni}{Auer
  et~al.}{2021}]{Auer:21}
\textsc{Auer, Simone, Marco Bernardini, and Martina Cecioni} (2021):
  \enquote{Corporate leverage and monetary policy effectiveness in the euro
  are,} \emph{European Economic Review}, 140, 103943.

\bibitem[\protect\citeauthoryear{Auerbach and Gorodnichenko}{Auerbach and
  Gorodnichenko}{2012}]{Auerbach:12}
\textsc{Auerbach, Alan~J. and Yuriy Gorodnichenko} (2012): \enquote{Measuring
  the Output Responses to Fiscal Policy,} \emph{American Economic Journal:
  Economic Policy}, 4 (2), 1--27.

\bibitem[\protect\citeauthoryear{Auerbach and Gorodnichenko}{Auerbach and
  Gorodnichenko}{2013}]{Auerbach:13b}
---\hspace{-.1pt}---\hspace{-.1pt}--- (2013): \enquote{Output spillovers from
  fiscal policy,} \emph{American Economic Review: Papers and Proceedings}, 103,
  141--146.

\bibitem[\protect\citeauthoryear{Auerbach and Gorodnichenko}{Auerbach and
  Gorodnichenko}{2016}]{Auerbach:16}
---\hspace{-.1pt}---\hspace{-.1pt}--- (2016): \enquote{Effects of fiscal shocks
  in a globalized world,} \emph{IMF Economic Review}, 54, 177--215.

\bibitem[\protect\citeauthoryear{{Ben~Zeev}, Ramey, and Zubairy}{{Ben~Zeev}
  et~al.}{2023}]{Ben-Zeev:23}
\textsc{{Ben~Zeev}, Nadav, Valerie~A. Ramey, and Sarah Zubairy} (2023):
  \enquote{Do Government Spending Multipliers Depend on the Sign of the Shock?}
  \emph{AEA Papers and Proceedings}, 113, 382--387.

\bibitem[\protect\citeauthoryear{Berge, Ridder, and Pfajfar}{Berge
  et~al.}{2021}]{Berge:21}
\textsc{Berge, Travis, Maarten~De Ridder, and Damjan Pfajfar} (2021):
  \enquote{When is the fiscal multiplier high? A comparison of four business
  cycle phases,} \emph{European Economic Review}, 138, 103852.

\bibitem[\protect\citeauthoryear{Bernardini and Peersman}{Bernardini and
  Peersman}{2018}]{Bernardini:18}
\textsc{Bernardini, Marco and Gert Peersman} (2018): \enquote{Private debt
  overhang and the government spending multiplier: Evidence for the United
  States,} \emph{Journal of Applied Econometrics}, 33, 485--508.

\bibitem[\protect\citeauthoryear{Bernardini, Schryder, and Peersman}{Bernardini
  et~al.}{2020}]{Bernardini:20}
\textsc{Bernardini, Marco, Selien~De Schryder, and Gert Peersman} (2020):
  \enquote{Heterogeneous government spending multipliers in the era surrounding
  the great recession,} \emph{Review of Economics and Statistics}, 102 (2),
  304--322.

\bibitem[\protect\citeauthoryear{Biolsi}{Biolsi}{2017}]{Biolsi:17}
\textsc{Biolsi, Christopher} (2017): \enquote{Nonlinear effects of fiscal
  policy over the business cycle,} \emph{Journal of Economic Dynamics and
  Control}, 78, 54--87.

\bibitem[\protect\citeauthoryear{Boehm}{Boehm}{2020}]{Boehm:20}
\textsc{Boehm, Christoph~E.} (2020): \enquote{Government consumption and
  investment: Does the composition of purchases affect the multiplier?}
  \emph{Journal of Monetary Economics}, 115, 80--93.

\bibitem[\protect\citeauthoryear{Born, M{\"u}ller, and Pfeifer}{Born
  et~al.}{2020}]{Born:20}
\textsc{Born, Benjamin, Gernot~J. M{\"u}ller, and Johannes Pfeifer} (2020):
  \enquote{Does austerity pay off?} \emph{Review of Economics and Statistics},
  102 (2), 323--338.

\bibitem[\protect\citeauthoryear{Cacciatore and Ravenna}{Cacciatore and
  Ravenna}{2021}]{Cacciatore:21}
\textsc{Cacciatore, Matteo and Federico Ravenna} (2021): \enquote{Uncertainty,
  Wages and the Business Cycle,} \emph{Economic Journal}, 131, 2797--2823.

\bibitem[\protect\citeauthoryear{Cattaneo, Crump, Farrell, and Feng}{Cattaneo
  et~al.}{2024}]{Cattaneo:24}
\textsc{Cattaneo, Matias~D., Richard~K. Crump, Max~H. Farrell, and Yingjie
  Feng} (2024): \enquote{On Binscatter,} \emph{American Economic Review}, 114
  (5), 1488--1514.

\bibitem[\protect\citeauthoryear{Choi, Shin, and Yoo}{Choi
  et~al.}{2022}]{Choi:22}
\textsc{Choi, Sangyup, Junghyeok Shin, and Seung~Yong Yoo} (2022): \enquote{Are
  government spending shocks inflationary at the zero lower bound? New evidence
  from daily data,} \emph{Journal of Economic Dynamics and Control}, 139,
  104423.

\bibitem[\protect\citeauthoryear{{De Haan} and Wiese}{{De Haan} and
  Wiese}{2022}]{Haan:22}
\textsc{{De Haan}, Jakob and Rasmus Wiese} (2022): \enquote{The impact of
  product and labour market reform on growth: Evidence for OECD countries based
  on local projections,} \emph{Journal of Applied Econometrics}, 37, 746--770.

\bibitem[\protect\citeauthoryear{Demirel}{Demirel}{2021}]{Demirel:21}
\textsc{Demirel, Ufuk~Devrim} (2021): \enquote{The short-term effects of tax
  changes: The role of state dependence,} \emph{Journal of Monetary Economics},
  117, 918--934.

\bibitem[\protect\citeauthoryear{Duval and Furceri}{Duval and
  Furceri}{2018}]{Duval:18}
\textsc{Duval, Romain and Davide Furceri} (2018): \enquote{The Effects of Labor
  and Product Market Reforms,} \emph{IMF Economic Review}, 66 (1), 31--69.

\bibitem[\protect\citeauthoryear{{El Herradi} and Leroy}{{El Herradi} and
  Leroy}{2021}]{El-Herradi:21}
\textsc{{El Herradi}, Mehdi and Aur{\'e}lien Leroy} (2021): \enquote{Monetary
  Policy and the Top 1\%: Evidence from a Century of Modern Economic History,}
  \emph{International Journal of Central Banking}, 18 (5), 237--277.

\bibitem[\protect\citeauthoryear{El-Shagi and von Schweinitz}{El-Shagi and von
  Schweinitz}{2021}]{El-Shagi:21}
\textsc{El-Shagi, Makram and Gregor von Schweinitz} (2021): \enquote{Fiscal
  policy and fiscal fragility: Empirical evidence from the OECD,} \emph{Journal
  of International Money and Finance}, 115, 102292.

\bibitem[\protect\citeauthoryear{Eminidou, Geiger, and Zachariadis}{Eminidou
  et~al.}{2023}]{Eminidou:23}
\textsc{Eminidou, Snezana, Martin Geiger, and Marios Zachariadis} (2023):
  \enquote{Public debt and state-dependent effects of fiscal policy in the euro
  area,} \emph{Journal of International Money and Finance}, 130, 102746.

\bibitem[\protect\citeauthoryear{Falck, Hoffmann, and H{\"u}rtgen}{Falck
  et~al.}{2021}]{Falck:21}
\textsc{Falck, Elisabeth, Mathias Hoffmann, and Patrick H{\"u}rtgen} (2021):
  \enquote{Disagreement about inflation expectations and monetary policy
  transmission,} \emph{Journal of Monetary Economics}, 118, 15--31.

\bibitem[\protect\citeauthoryear{Fan and Zhang}{Fan and Zhang}{2008}]{Fan:08}
\textsc{Fan, Jianqing and Wenyang Zhang} (2008): \enquote{Statistical methods
  with varying coefficient models,} \emph{Statistics and Its Interface}, 1,
  179--195.

\bibitem[\protect\citeauthoryear{Furceri, Loungani, and Zdzienicka}{Furceri
  et~al.}{2018}]{Furceri:18}
\textsc{Furceri, Davide, Prakash Loungani, and Aleksandra Zdzienicka} (2018):
  \enquote{The effects of monetary policy shocks on inequality,} \emph{Journal
  of International Money and Finance}, 85, 168--186.

\bibitem[\protect\citeauthoryear{Ghassibe and Zanetti}{Ghassibe and
  Zanetti}{2022}]{Ghassibe:22}
\textsc{Ghassibe, Mishel and Francesco Zanetti} (2022): \enquote{State
  dependence of fiscal multipliers: the source of fluctuations matters,}
  \emph{Journal of Monetary Economics}, 132, 1--23.

\bibitem[\protect\citeauthoryear{Gon{\c c}alves, Herrera, Kilian, and
  Pesavento}{Gon{\c c}alves et~al.}{2024}]{Goncalves:24}
\textsc{Gon{\c c}alves, S{\'\i}lvia, Ana~Mar{\'\i}a Herrera, Lutz Kilian, and
  Elena Pesavento} (2024): \enquote{State-dependent local projections,}
  \emph{Journal of Econometrics}, 244 (2), 105702.

\bibitem[\protect\citeauthoryear{Hansen}{Hansen}{2022}]{Hansen:22}
\textsc{Hansen, Bruce~E.} (2022): \emph{Econometrics}, Princeton: Princeton
  University Press.

\bibitem[\protect\citeauthoryear{Jord{\`a}, Schularick, and Taylor}{Jord{\`a}
  et~al.}{2020}]{Jorda:20}
\textsc{Jord{\`a}, {\`O}scar, Moritz Schularick, and Alan~M. Taylor} (2020):
  \enquote{The effects of quasi-random monetary experiments,} \emph{Journal of
  Monetary Economics}, 112, 22--40.

\bibitem[\protect\citeauthoryear{Jord{\`a} and Taylor}{Jord{\`a} and
  Taylor}{2016}]{Jorda:16}
\textsc{Jord{\`a}, {\`O}scar and Alan~M. Taylor} (2016): \enquote{The Time for
  Austerity: Estimating the Average Treatment Effect of Fiscal Policy,}
  \emph{Economic Journal}, 126, 219--255.

\bibitem[\protect\citeauthoryear{Klein}{Klein}{2017}]{Klein:17}
\textsc{Klein, Mathias} (2017): \enquote{Austerity and Private Debt,}
  \emph{Journal of Money, Credit and Banking}, 49 (7), 1555--1585.

\bibitem[\protect\citeauthoryear{Klein, Polattimur, and Winkler}{Klein
  et~al.}{2022}]{Klein:22}
\textsc{Klein, Mathias, Hamza Polattimur, and Roland Winkler} (2022):
  \enquote{Fiscal spending multipliers over the household leverage cycle,}
  \emph{European Economic Review}, 141, 103989.

\bibitem[\protect\citeauthoryear{Klein and Winkler}{Klein and
  Winkler}{2021}]{Klein:21}
\textsc{Klein, Mathias and Roland Winkler} (2021): \enquote{The government
  spending multiplier at the zero lower bound: International evidence from
  historical data,} \emph{Journal of Applied Econometrics}, 36, 744--759.

\bibitem[\protect\citeauthoryear{Lastauskas and Stak{\.e}nas}{Lastauskas and
  Stak{\.e}nas}{2020}]{Lastauskas:20}
\textsc{Lastauskas, Povilas and Julius Stak{\.e}nas} (2020): \enquote{Labor
  market reforms and the monetary policy environment,} \emph{European Economic
  Review}, 128, 103509.

\bibitem[\protect\citeauthoryear{Leduc and Wilson}{Leduc and
  Wilson}{2012}]{Leduc:12}
\textsc{Leduc, Sylvain and Daniel Wilson} (2012): \enquote{Roads to Prosperity
  or Bridges to Nowhere? Theory and Evidence on the Impact of Public
  Infrastructure Investment,} \emph{NBER Macroeconomics Annual}, 27, 89--142.

\bibitem[\protect\citeauthoryear{Liu}{Liu}{2022}]{Liu:22}
\textsc{Liu, Siming} (2022): \enquote{Government spending during sudden stop
  crises,} \emph{Journal of International Economics}, 135, 103571.

\bibitem[\protect\citeauthoryear{Liu}{Liu}{2023}]{Liu:23}
\textsc{Liu, Yang} (2023): \enquote{Government debt and risk premia,}
  \emph{Journal of Monetary Economics}, 136, 18--34.

\bibitem[\protect\citeauthoryear{Loipersberger and Matschke}{Loipersberger and
  Matschke}{2022}]{Loipersberger:22}
\textsc{Loipersberger, Florian and Johannes Matschke} (2022):
  \enquote{Financial cycles and domestic policy choices,} \emph{European
  Economic Review}, 143, 104034.

\bibitem[\protect\citeauthoryear{L{\"u}tkepohl}{L{\"u}tkepohl}{2005}]{Lutkepohl:05}
\textsc{L{\"u}tkepohl, Helmut} (2005): \emph{New Introduction to Multiple Time
  Series Analysis}, Berlin: Springer.

\bibitem[\protect\citeauthoryear{Miyamoto, Nguyen, and Sergeyev}{Miyamoto
  et~al.}{2018}]{Miyamoto:18}
\textsc{Miyamoto, Wataru, Thuy~Lan Nguyen, and Dmitriy Sergeyev} (2018):
  \enquote{Government Spending Multipliers under the Zero Lower Bound: Evidence
  from Japan,} \emph{American Economic Journal: Macroeconomics}, 10 (3),
  247--277.

\bibitem[\protect\citeauthoryear{Owyang, Ramey, and Zubairy}{Owyang
  et~al.}{2013}]{Owyang:13}
\textsc{Owyang, Michael~T., Valerie~A. Ramey, and Sarah Zubairy} (2013):
  \enquote{Are Government Spending Multipliers Greater during Periods of Slack?
  Evidence from Twentieth-Century Historical Data,} \emph{American Economic
  Review: Papers and Proceedings}, 103 (3), 129--134.

\bibitem[\protect\citeauthoryear{Ramey and Zubairy}{Ramey and
  Zubairy}{2018}]{Ramey:18}
\textsc{Ramey, Valerie~A. and Sarah Zubairy} (2018): \enquote{Government
  Spending Multipliers in Good Times and in Bad: Evidence from US Historical
  Data,} \emph{Journal of Political Economy}, 126 (2), 850--901.

\bibitem[\protect\citeauthoryear{Riera-Crichton, Vegh, and
  Vuletin}{Riera-Crichton et~al.}{2015}]{Riera:15}
\textsc{Riera-Crichton, Daniel, Carlos~A. Vegh, and Guillermo Vuletin} (2015):
  \enquote{Procyclical and countercyclical fiscal multipliers: Evidence from
  OECD countries,} \emph{Journal of International Money and Finance}, 52,
  15--31.

\bibitem[\protect\citeauthoryear{Santoro, Petrella, Pfajfar, and
  Gaffeo}{Santoro et~al.}{2014}]{Santoro:14}
\textsc{Santoro, Emiliano, Ivan Petrella, Damjan Pfajfar, and Edoardo Gaffeo}
  (2014): \enquote{Loss aversion and the asymmetric transmission of monetary
  policy,} \emph{Journal of Monetary Economics}, 68, 19--36.

\bibitem[\protect\citeauthoryear{Sheng and Sukaj}{Sheng and
  Sukaj}{2021}]{Sheng:21}
\textsc{Sheng, Xuguang~Simon and Rubena Sukaj} (2021): \enquote{Identifying
  external debt shocks in low- and middle-income countries,} \emph{Journal of
  International Money and Finance}, 110, 102283.

\bibitem[\protect\citeauthoryear{Sheremirov and Spirovska}{Sheremirov and
  Spirovska}{2022}]{Sheremirov:22}
\textsc{Sheremirov, Viacheslav and Sandra Spirovska} (2022): \enquote{Fiscal
  multipliers in advanced and developing countries: Evidence from military
  spending,} \emph{Journal of Public Economics}, 208, 104631.

\bibitem[\protect\citeauthoryear{Stock and Watson}{Stock and
  Watson}{2016}]{Stock:16}
\textsc{Stock, James~H. and Mark~W. Watson} (2016): \enquote{Dynamic Factor
  Models, Factor-Augmented Vector Autoregressions, and Structural Vector
  Autoregressions in Macroeconomics,} in \emph{Handbook of Macroeconomics}, ed.
  by John~B. Taylor and Harald Uhlig, Elsevier, vol.~2A, 415--525.

\bibitem[\protect\citeauthoryear{Tenreyro and Thwaites}{Tenreyro and
  Thwaites}{2016}]{Tenreyro:16}
\textsc{Tenreyro, Silvana and Gregory Thwaites} (2016): \enquote{Pushing on a
  String: US Monetary Policy is Less Powerful in Recessions,} \emph{American
  Economic Journal: Macroeconomics}, 8 (4), 43--74.

\bibitem[\protect\citeauthoryear{Tillmann}{Tillmann}{2020}]{Tillmann:20}
\textsc{Tillmann, Peter} (2020): \enquote{Monetary Policy Uncertainty and the
  Response of the Yield Curve to Policy Shocks,} \emph{Journal of Money, Credit
  and Banking}, 52 (4), 803--833.

\bibitem[\protect\citeauthoryear{Virolainen}{Virolainen}{2026}]{sstvars}
\textsc{Virolainen, Savi} (2026): \emph{sstvars: Toolkit for Reduced Form and
  Structural Smooth Transition Vector Autoregressive Models}, University of
  Helsinki, r package version 1.2.5.

\end{thebibliography}
\end{document}